\documentclass[hyperref,final]{llncs}

\usepackage[utf8]{inputenc}
\usepackage{lmodern}
\usepackage{a4wide}

\usepackage{amsfonts,amssymb,amsmath}
\usepackage{enumerate}

\usepackage{xspace}
\usepackage{xparse}
\usepackage{xcolor,colortbl}
\usepackage{hhline}
\usepackage[basic]{complexity}
\renewclass\P{PTIME}
\renewclass\EXP{EXPTIME}
\renewclass\coNEXP{coNEXPTIME}
\newclass\coNPSPACE{coNPSPACE}
\usepackage{tikz}
\usetikzlibrary{arrows,automata,shapes,calc,positioning,intersections,backgrounds}
\usepackage{wrapfig}

\tikzstyle{player1}=[draw,rounded rectangle, minimum size=7mm]
\tikzstyle{player2}=[draw,rectangle,minimum size=7mm]
\tikzstyle{widget}=[draw,ellipse,dashed,minimum size=6mm]
\tikzset{every loop/.style={looseness=7}}

\usepackage{paralist}
\usepackage[ruled,boxed,vlined,linesnumbered]{algorithm2e}

\usepackage{url}

\usepackage{hyperref}

\usepackage{graphicx}

\newcommand\N{\ensuremath{\mathbb{N}}\xspace}
\renewcommand\R{\ensuremath{\mathbb{R}}\xspace}
\newcommand\Rplus{\ensuremath{\R_{\geq 0}}\xspace}
\newcommand\Rbar{\ensuremath{\R_\infty}\xspace}
\newcommand\Z{\ensuremath{\mathbb{Z}}\xspace}

\newcommand\Q{\ensuremath{\mathbb{Q}}\xspace}

\renewcommand\leq{\leqslant}
\renewcommand\geq{\geqslant}

\newcommand\powerset[1]{\ensuremath{2^{#1}}\xspace}

\newcommand\WTG{WTG\xspace}

\newcommand\Clocks{\ensuremath{X}\xspace}
\newcommand\clockbound{\ensuremath{M}\xspace}
\newcommand\val{\ensuremath{\nu}\xspace}
\newcommand\valnull{\ensuremath{\mathbf{0}}\xspace}
\newcommand\guard{\ensuremath{g}\xspace}
\newcommand\Guards{\ensuremath{\mathsf{Guards}(\Clocks)}\xspace}
\newcommand\reset{\ensuremath{Y}\xspace}

\newcommand\MinPl{\ensuremath{\mathsf{Min}}\xspace}
\newcommand\MaxPl{\ensuremath{\mathsf{Max}}\xspace}

\newcommand\game{\ensuremath{\mathcal G}\xspace}

\newcommand\tgame{\ensuremath{{\mathcal T(\game)}}\xspace}

\newcommand\loc{\ensuremath{\ell}\xspace}
\newcommand\Locs{\ensuremath{L}\xspace}
\newcommand\LocsMin{\ensuremath{\Locs_{\MinPl}}\xspace}
\newcommand\LocsMax{\ensuremath{\Locs_{\MaxPl}}\xspace}

\newcommand\Trans{\ensuremath{\Delta}\xspace}
\newcommand\trans{\ensuremath{\delta}\xspace}

\newcommand\weight{\ensuremath{\mathsf{wt}}\xspace}
\newcommand\LocsT{\ensuremath{\Locs_T}\xspace}
\newcommand\weightT{\ensuremath{\weight_T}\xspace}
\newcommand\weightC{\ensuremath{\weight_\Sigma}\xspace}

\newcommand\wmax{w_{\mathrm{max}}\xspace}

\newcommand\myconst{\mathbf{B}\xspace}

\newcommand\play{\ensuremath{\rho}\xspace}

\newcommand\FPlays{\ensuremath{\mathsf{FPlays}}\xspace}
\newcommand\FPlaysMin{\FPlays^\MinPl}
\newcommand\FPlaysMax{\FPlays^\MaxPl}

\newcommand\first{\ensuremath{\mathsf{first}}\xspace}
\newcommand\last{\ensuremath{\mathsf{last}}\xspace}

\newcommand\strat{\ensuremath{\sigma}\xspace}
\newcommand\stratmin{\ensuremath{\strat_{\MinPl}}\xspace}
\newcommand\stratmax{\ensuremath{\strat_{\MaxPl}}\xspace}

\newcommand\minstrategy{\stratmin}
\newcommand\maxstrategy{\stratmax}

\newcommand\Play{\ensuremath{\mathsf{play}}\xspace}

\newcommand\outcomes{\Play}

\newcommand\Val{\ensuremath{\mathsf{Val}}\xspace}

\newcommand\rgame{\ensuremath{\mathcal{R}(\game)}\xspace}
\newcommand\Nrgame[1]{\ensuremath{\mathcal{R}_{#1}(\game)}\xspace}

\newcommand\rstate{\ensuremath{s}\xspace}
\newcommand\RStates{\ensuremath{S}\xspace}

\newcommand\RStatesF{\ensuremath{\RStates_f}\xspace}
\newcommand\RStatesK{\ensuremath{\RStates_\Kernel}\xspace}

\newcommand\rtrans{\ensuremath{t}\xspace}
\newcommand\RTrans{\ensuremath{T}\xspace}
\newcommand\RTransF{\ensuremath{\RTrans_f}\xspace}
\newcommand\RTransK{\ensuremath{\RTrans_\Kernel}\xspace}

\newcommand\rpath{\ensuremath{\pi}\xspace} %

\newcommand\regions[2]{\mathsf{Reg}(#1,#2)}
\newcommand\Nregions[3]{\mathsf{Reg}_{#1}(#2,#3)}

\newcommand\corv{\ensuremath{v}\xspace}

\newcommand\Ncgame[1]{\ensuremath{{\cal C}_{#1}(\game)}\xspace}

\newcommand\lipconst{\ensuremath{\Lambda}\xspace}
\newcommand\numpieces{\ensuremath{J}\xspace}
\newcommand\addconstbound{\ensuremath{U}\xspace}

\newcommand\Kernel{\ensuremath{{\mathsf{K}}}\xspace}

\renewcommand\paragraph[1]{\smallskip\noindent\textbf{#1.}}

 \title{Symbolic Approximation of Weighted Timed Games\thanks{This work has been funded by the DeLTA project
   (ANR-16-CE40-0007).}}

 \author{Damien~Busatto-Gaston, Benjamin~Monmege and
   Pierre-Alain~Reynier}
 \institute{Aix Marseille Univ, Universit\'e de
   Toulon, CNRS, LIS, Marseille,
   France\\
 \email{\{damien.busatto,benjamin.monmege,pierre-alain.reynier\}@lis-lab.fr}}

\begin{document}

\maketitle

\begin{abstract}
  Weighted timed games are zero-sum games played by two players on a
  timed automaton equipped with weights, where one player wants to
  minimise the accumulated weight while reaching a target.
  Weighted timed games are notoriously difficult and quickly
  undecidable, even when restricted to non-negative weights.
  For non-negative weights, the largest class that can be analysed
  has been introduced by Bouyer, Jaziri and Markey in 2015.
  Though the value problem is undecidable, the authors show
  how to approximate the value by considering regions with a refined
  granularity.
  In this work, we extend this class to incorporate negative weights,
  allowing one to model energy for instance, and prove that the value
  can still be approximated, with the same complexity.
  In addition, we show that a symbolic
  algorithm, relying on the paradigm of value iteration, can be used
  as an approximation schema on this class.
\end{abstract}

\section{Introduction}

The design of programs verifying some real-time specifications is a
notoriously difficult problem, because such programs must take care of
delicate timing issues, and are difficult to debug a posteriori. One
research direction to ease the design of real-time software is to
automatise the process. The situation may be modelled into a
timed game, played by a \emph{controller} and an antagonistic
\emph{environment}: they act, in a turn-based fashion, over a
\emph{timed automaton}~\cite{AD94}, namely a finite automaton equipped
with real-valued variables, called clocks, evolving with a uniform
rate. A simple, yet realistic, objective for the controller is to
reach a target location. We are thus looking for a \emph{strategy} of the
controller, that is a recipe dictating
how to play %
so that the target is reached no matter how
the environment plays. Reachability timed games are
decidable~\cite{AsaMal99}, and \EXP-complete~\cite{JurTri07}.

Weighted extensions of these games have been considered in order to
measure the quality of the winning strategy for the
controller~\cite{BCFL04,ABM04}: when the controller has several winning
strategies in a given reachability timed game, the quantitative version
of the game helps choosing a good one with respect to some
metrics. This means that the game now takes place over a
\emph{weighted (or priced) timed automaton}~\cite{BehFeh01,AluLa-04},
where transitions are equipped with weights, and locations with rates of
weights (the cost is then proportional to the time spent in this
location, with the rate as proportional coefficient).
While solving the optimal reachability problem on
weighted timed automata has been shown to be
\PSPACE-complete~\cite{BouBri07} (i.e.~the same complexity as the
non-weighted version), weighted timed games are known to be
undecidable~\cite{BBR05}. This has led to many restrictions in order
to regain decidability, the first and most interesting one being the
class of strictly non-Zeno cost with only non-negative weights (in
transitions and locations)~\cite{BCFL04}: this hypothesis requires
that every execution of the timed automaton that follows a cycle of
the region automaton has a weight far from 0 (in interval
$[1,+\infty)$, for instance).

Negative weights are crucial when one
wants to model energy or other resources that can grow or decrease
during the execution of the system to study.
In \cite{BMR17a}, we have recently extended the strictly non-Zeno cost restriction to
weighted timed games in the presence of negative weights in
transitions and/or locations.
We have described there the class of \emph{divergent weighted timed games} where each
execution that follows a cycle of the region automaton has a weight
far from 0, i.e.~in $(-\infty,-1]\cup[1,+\infty)$. We were able to
obtain a doubly-exponential-time algorithm to compute the values and
almost-optimal strategies, while deciding the divergence of a weighted
timed game is \PSPACE-complete. These complexity results match the
ones that could be obtained in the non-negative case from~\cite{BCFL04,ABM04}.

The techniques used to obtain the results of \cite{BMR17a} cannot be extended if the
conditions are slightly relaxed. For instance, if we add the
possibility for an execution of the timed automaton following a cycle
of the region automaton to have weight \emph{exactly 0}, the decision
problem is known to be undecidable~\cite{BJM15}, even with
non-negative weights only. For this extension, in the presence of
non-negative weights only, it has been proposed an approximation
schema
 to compute arbitrarily close estimates of the optimal value
\cite{BJM15}.
To this end, the authors consider regions with a refined granularity so as
to control the precision of the approximation.
In this work, our contribution
is two-fold: first, we extend the class considered in \cite{BJM15}
to the presence of negative weights; second, we show that
the approximation can be obtained using a symbolic computation,
based on the paradigm of value iteration.

More precisely, we define the class of \emph{almost-divergent weighted
  timed games} where, for each strongly connected component (SCC) of
the region automaton, executions following a cycle of this SCC have
weights either all in $(-\infty,-1]\cup \{0\}$, or all in
$\{0\}\cup[1,+\infty)$. In contrast, the \emph{divergent} condition is
equivalent to the same property on the strongly connected components,
but without the presence of singleton $\{0\}$. Given an
almost-divergent weighted timed game, an initial configuration $c$ and
a threshold~$\varepsilon$, we compute a value that we guarantee to be
$\varepsilon$-close to the optimal value when the play starts from
$c$.
 Moreover, we prove that deciding if a weighted timed game is
 almost-divergent is a \PSPACE-complete problem.

In order to approximate almost-divergent weighted timed games,
we first adapt the approximation schema of \cite{BJM15} to our setting.
At the very core of their schema is the notion of \emph{kernels}
that collect all cycles of weight exactly 0 in the game. Then, a semi-unfolding
of the game (in which kernels are not unfolded) of bounded depth is shown
to be equivalent to the original game. Adapting this schema to negative weights
requires to address new issues:
\begin{itemize}
\item The definition and the approximation of these kernels is much more
  intricate in our setting (see Sections~\ref{sec:kernels-in-wtg} and
  \ref{sec:approx-kernels}). Indeed, with only non-negative weights, a
  cycle of weight $0$ only encounters locations and transitions with
  weight $0$. It is no longer the case with arbitrary weights,
  both for discrete
  weights on transitions (that could alternate between weight $+1$ and
  $-1$, e.g.)  and continuous rates on locations: for this continuous
  part, this requires to keep track of the real-time dynamics of the game.
\item Some configurations may have value $-\infty$. While it is
  undecidable in general whether a configuration has value $-\infty$,
  we prove that it is decidable for almost-divergent weighted timed
  games (see Lemma~\ref{lm:-infty}).
\item The identification of an adequate bound to define an equivalent
  semi-unfolding of bounded depth is more difficult in our setting, as
  having guarantees on weight accumulation is harder
  (we can lose accumulated weight). We deal with this by
  evaluating how large the value of a configuration can be,
  provided it is not infinite. This is presented in
  Section~\ref{sec:unfolding}.
\end{itemize}

We also develop, in Section~\ref{sec:symbolic}, a more symbolic approximation
schema, in the sense that it
avoids %
the a priori refinement of regions. Instead, all
computations are performed in a symbolic way using the techniques
developed in~\cite{ABM04}.
This allows to mutualise as much as possible the different
computations: comparing these schemas with the evaluation of MDPs or
quantitative games like mean-payoff or discounted-payoff, it is the
same improvement as when using \emph{value iteration} techniques
instead of techniques based on the unfolding of the model into a finite tree
which can contain many times the same location.

\section{Weighted timed games}\label{sec:prelim}

\paragraph{Clocks, guards and regions}
We let \Clocks be a finite set of variables called clocks. A valuation of
clocks is a mapping $\val\colon \Clocks\to \Rplus$. For a valuation $\val$,
$d\in\Rplus$ and $Y\subseteq \Clocks$, we define the valuation $\val+d$ as
$(\val+d)(x)=\val(x)+d$, for all $x\in \Clocks$, and the valuation
$\val[Y:=0]$ as $(\val[Y:=0])(x)=0$ if $x\in Y$, and
$(\val[Y:=0])(x)=\val(x)$ otherwise. The valuation $\valnull$
assigns $0$ to every clock.
A guard on clocks of \Clocks is a conjunction of atomic constraints of
the form $x\bowtie c$, where ${\bowtie}\in\{{\leq},<,=,>,{\geq}\}$ and
$c\in \Q$ (we allow for rational coefficients as we will
refine the granularity in the following). Guard $\overline g$ is the closed
version of a satisfiable guard $g$ where every open constraint $x<c$ or $x>c$ is
replaced by its closed version $x\leq c$ or $x\geq c$. A valuation
$\val\colon \Clocks\to \Rplus$ satisfies an atomic constraint
$x\bowtie c$ if $\val(x)\bowtie c$. The satisfaction relation is
extended to all guards $g$ naturally, and denoted by $\val\models
g$. We let $\Guards$ denote the set of guards over \Clocks.

\begin{wrapfigure}{3}{3cm}
  \centering
  \vspace{-5mm}
  \begin{tikzpicture}
    \path[draw,color=gray!60!white,very thin] (0,1.66) -- (0.33,2) -- (0.33,0) --
    (2,1.66) -- (0,1.66)%
    (0,1.33) -- (0.66,2) -- (0.66,0) -- (2,1.33) -- (0,1.33)%
    (0,0.66) -- (1.33,2) -- (1.33,0) -- (2,0.66) -- (0,0.66)%
    (0,0.33) -- (1.66,2) -- (1.66,0) -- (2,0.33) -- (0,0.33);%

    \path[draw,->,thick](0,0) -> (2.3,0) node[above] {$x$};
    \path[draw,->,thick](0,0) -> (0,2.3) node[right] {$y$};

    \path[draw] (0,0) -- (2,2)%
    (1,0) -- (1,2) -- (0,1) -- (2,1) -- (1,0)%
    (0,2) -- (2,2) -- (2,0);%

    \node () at (1,-0.2) {$1$};
    \node () at (2,-0.2) {$2$};
    \node () at (-0.15,-0.15) {$0$};
    \node () at (-0.2,1) {$1$};
    \node () at (-0.2,2) {$2$};

  \end{tikzpicture}
    \vspace{-8mm}
\end{wrapfigure}
We rely on the crucial notion of regions, as introduced in the seminal
work on timed automata \cite{AD94}: intuitively, a region is a set of valuations
that are all time-abstract bisimilar. We will need some
refinement of regions, with respect to a granularity $1/N$, with
$N\in \N$.
Formally, with respect to the set \Clocks of clocks and
a constant $M$, a $1/N$-region $r$ is a subset of valuations
characterised by the vector
$(\iota_x)_{x\in \Clocks} = (\min(M N, \lfloor \val(x)
N\rfloor))_{x\in \Clocks}\in [0,MN]^\Clocks$ and the order of
fractional parts of $\val(x) N$, given as a partition
$\Clocks=\Clocks_0\uplus \Clocks_1\uplus \cdots \uplus \Clocks_m$ of
clocks: a valuation~$\val$ is in this $1/N$-region~$r$ if
\begin{inparaenum}[($i$)]
\item $\lfloor \val(x) N\rfloor = \iota_x$, for all clocks
  $x\in \Clocks$;
\item $\val(x)=0$ for all $x\in \Clocks_0$;
\item all clocks $x\in \Clocks_i$ satisfy that
  $\val(x) N$ have the same fractional part, for all $1\leq i\leq m$.
\end{inparaenum}
We denote by $\Nregions N \Clocks \clockbound$ the set of
$1/N$-regions, and we write $\regions \Clocks \clockbound$ as a
shorthand for $\Nregions 1 \Clocks \clockbound$. We recover the
traditional notion of region for $N=1$. E.g., the figure on the right
depicts regions $\regions{\{x,y\}} 2$ as well as their refinement
$\Nregions 3 {\{x,y\}} 2$. For any integer guard $g$, either all
valuations of a given $1/N$-region satisfy $g$, or none of them do.
A $1/N$-region $r'$ is said to be a time successor of the $1/N$-region
$r$ if there exist $\val\in r$, $\val'\in r'$, and $d>0$ such that
$\val'=\val+d$. Moreover, for $Y\subseteq \Clocks$, we let
$r[Y:=0]$ be the $1/N$-region where clocks of $Y$ are~reset.

\begin{figure}
  \centering
  \scalebox{0.9}{
     \begin{tikzpicture}[node distance=3cm,auto,->,>=latex]

      \node[player2](1){\makebox[0mm][c]{$\mathbf{-2}$}};
      \node()[below of=1,node distance=6mm]{$\loc_1$};

      \node[player1](2)[below right
        of=1]{\makebox[0mm][c]{$\mathbf{2}$}};
      \node()[below of=2,node distance=6mm]{$\loc_2$};

      \node[player1](3)[above right
        of=2, accepting]{\makebox[0mm][c]{}};
      \node()[below of=3,node distance=6mm]{$\loc_3$};
    \node()[above of=3,node distance=6mm]{$\weightT=\mathbf{0}$};

      \node[player2](4)[below right
        of=3]{\makebox[0mm][c]{$\mathbf{-1}$}};
      \node()[below of=4,node distance=6mm]{$\loc_4$};

      \node[player1](5)[above right
        of=4]{\makebox[0mm][c]{$\mathbf{-2}$}};
      \node()[below of=5,node distance=6mm]{$\loc_5$};

	\path
	(2) edge node[below left,xshift=2mm,yshift=2mm]{$
          \begin{array}{c}
x\leq2 \\ x:=0 \\ \mathbf{0}
          \end{array}
          $} (1);

	\path
	(2) edge node[above left,xshift=4mm,yshift=-4mm]{
          $\begin{array}{c}
1\leq x\leq3 \\ \mathbf{1}
          \end{array}$} (3);

	\path
	(2) edge node[below]{$x\leq3;\;  x:=0;\; \mathbf{0}$} (4);

	\path
	(4) edge node[above right,xshift=-4mm,yshift=-4mm]{$
          \begin{array}{c}
2\leq x\leq3 \\ \mathbf{3}
          \end{array}$} (3);

	\path
	(4) edge node[below right,xshift=-2mm,yshift=2mm]{$
          \begin{array}{c}
x\leq3\\ \mathbf{0}
          \end{array}$} (5);

	\path
	(1) edge node[above]{$x\leq3;\; \mathbf{0}$} (3);

	\path
	(5) edge node[above]{$x\leq3; \; \mathbf{0}$} (3);

	\path
	(1) edge [loop above] node {$
          \begin{array}{c}
x\leq 1\\ x:=0;\; \mathbf{3}
          \end{array}$}  (1);

	\path
	(5) edge [loop above] node {$\begin{array}{c}1<x\leq 3 \\ x:= 0; \;\mathbf{1}\end{array}$} (5);

      \end{tikzpicture}}
    \scalebox{.8}{\begin{tikzpicture}[node distance=3cm,auto,>=latex]

	\draw[->] (0,0) -- (0,4);
	\draw[->] (0,0) -- (4,0);

	\draw[dashed] (1,0) -- (1,4);
	\draw[dashed] (2,0) -- (2,4);
	\draw[dashed] (3,0) -- (3,4);

	\draw[dashed] (0,1) -- (4,1);
	\draw[dashed] (0,2) -- (4,2);
	\draw[dashed] (0,3) -- (4,3);

	\draw[dashed] (.666,0) node[below,xshift=-1mm]{$2/3$}-- (.666,1.666);

	\node at (4,-.5) {$x$};
  \node at (-.5,4) {$\Val$};

	\node at (0,-.5) {$0$};
	\node at (1,-.5) {$1$};
	\node at (2,-.5) {$2$};
	\node at (3,-.5) {$3$};

	\node at (-.5,0) {$0$};
	\node at (-.5,1) {$1$};
	\node at (-.5,2) {$2$};
	\node at (-.5,3) {$3$};

	\draw[red,very thick] (0,3) -- (1,1) -- (3,1);

	\draw[blue,very thick] (0,1) -- (2,3) -- (3,3);

	\node[rectangle,fill=white,fill opacity=.9,text opacity=1] at (2.4,3.4) {$\loc_2 \rightarrow \loc_4 \rightarrow \loc_3$};

	\node[rectangle,fill=white,fill opacity=.9,text opacity=1]  at (2.2,1.4) {$\loc_2 \rightarrow \loc_3$};

   \end{tikzpicture}}
 \caption{On the left, a weighted timed game. Locations belonging to \MinPl
   (resp.~\MaxPl) are depicted by circles (resp.~squares). The target
   location is $\loc_3$, whose output weight function is null.  It is
   easy to observe that location $\loc_1$ (resp.~$\loc_5$) has value
   $+\infty$ (resp.~$-\infty$). As a consequence, the value
   in~$\loc_4$ is determined by the edge to $\loc_3$, and depicted in
   blue on the right. In location $\loc_2$, the value
   associated with the transition to $\loc_3$ is depicted in red, and
   the value in $\loc_2$ is obtained as the minimum of these two
   curves. Observe the intersection in $x=2/3$ requiring to refine the
   regions.}\label{fig:wtg}
\end{figure}

\paragraph{Weighted timed games}
A weighted timed game (\WTG) is then a tuple
$\game=\langle\Locs=\LocsMin\uplus\LocsMax, \Trans, \weight, \LocsT,
\weightT\rangle$ where $\LocsMin$ and $\LocsMax$ are finite disjoint
subsets of locations belonging to \MinPl and \MaxPl, respectively,
$\Trans\subseteq \Locs\times\Guards\times \powerset \Clocks \times
\Locs$ is a finite set of transitions,
$\weight\colon \Trans\uplus\Locs \to \Z$ is the weight function,
associating an integer weight with each transition and location,
$\LocsT\subseteq \LocsMin$ is a subset of target locations for player
$\MinPl$, and $\weightT\colon \LocsT\times\Rplus^\Clocks\to \Rbar$ is
a function mapping each target location and valuation of the clocks to
a final weight of $\Rbar= \R\uplus\{-\infty,+\infty\}$ (possibly $0$,
$+\infty$, or $-\infty$). The addition of target weights is not
standard, but we will use it in the process of solving those games:
anyway, it is possible to simply map each target location to the
weight $0$, allowing us to recover the standard definition. Without
loss of generality, we suppose the absence of deadlocks except on
target locations, i.e.~for each location
$\loc\in \Locs\backslash\LocsT$ and valuation $\val$, there exists
$(\loc,g,Y,\loc')\in \Trans$ such that $\val\models g$, and no
transitions start in \LocsT.

\begin{wrapfigure}{r}{3.7cm}
   \vspace{-.5cm}
\scalebox{.7}{

   }
   \vspace{-.5cm}
\end{wrapfigure}
The semantics of a \WTG $\game$ is defined in terms of a game played
on an infinite transition system whose vertices are configurations of
the \WTG. A configuration is a pair $(\loc,\val)$ with a location and a
valuation of the clocks. Configurations are split into players
according to the location. A configuration is final if its location is a
target location of $\LocsT$. The alphabet of the transition system is
given by $\Rplus\times \Trans$ and will encode the delay that a player
wants to spend in the current location, before firing a certain
transition. For every delay $d\in\Rplus$, transition
$\trans=(\loc,g,Y,\loc')\in \Trans$ and valuation~$\val$, there is an
edge $(\loc,\val)\xrightarrow{d,\trans}(\loc',\val')$ if
$\val+d\models g$ and $\val'=(\val+d)[Y:=0]$. The weight of
such an edge $e$ is given by
$d\times \weight(\loc) + \weight(\trans)$.
An example is depicted on Figure~\ref{fig:wtg}.

A \emph{finite play} is a finite sequence of consecutive edges
$\play=(\loc_0,\val_0)\xrightarrow{d_0,\trans_0}(\loc_1,\val_1)
\xrightarrow{d_1,\trans_1}\cdots
(\loc_k,\val_k)$. We denote by $|\play|$ the length $k$ of $\play$.
The concatenation of two finite plays $\play_1$ and
$\play_2$, such that $\play_1$ ends in the same configuration as
$\play_2$ starts, is denoted by $\play_1\play_2$.
We let $\FPlays_\game$ be the set of all finite plays in $\game$, whereas
$\FPlaysMin_\game$ (resp. $\FPlaysMax_\game$) denote the finite plays that
end in a configuration of $\MinPl$ (resp. $\MaxPl$). A \emph{play}
is then a maximal sequence of consecutive edges (it is either infinite or it reaches \LocsT).

A \emph{strategy} for $\MinPl$ (resp.~$\MaxPl$) is a mapping
$\strat\colon \FPlaysMin_\game \to \Rplus\times \Trans$
(resp.~$\strat\colon \FPlaysMax_\game \to \Rplus\times \Trans$) such
that for all finite plays $\play\in\FPlaysMin_\game$
(resp.~$\play\in\FPlaysMax_\game$) ending in non-target configuration
$(\loc,\val)$, there exists an edge
$(\loc,\val)\xrightarrow{\strat(\play)}(\loc',\val')$. A play or
finite play
$\play = (\loc_0,\val_0)\xrightarrow{d_0,\trans_0}
(\loc_1,\val_1)\xrightarrow{d_1,\trans_1}\cdots$ conforms to a
strategy $\strat$ of $\MinPl$ (resp.~$\MaxPl$) if for all $k$ such
that $(\loc_k,\val_k)$ belongs to $\MinPl$ (resp.~$\MaxPl$), we have
that
$(d_{k},\trans_k) = \strat((\loc_0,\val_0)\xrightarrow{d_0,\trans_0}
\cdots (\loc_k,\val_k))$. A
strategy $\strat$ is \emph{memoryless} if for all finite plays
$\play, \play'$ ending in the same configuration, we have that
$\strat(\play)=\strat(\play')$.  For all strategies $\minstrategy$ and
$\maxstrategy$ of players \MinPl and \MaxPl, respectively, and for all
configurations~$(\loc_0,\val_0)$, we let
$\outcomes_\game((\loc_0,\val_0),\maxstrategy,\minstrategy)$ be the
outcome of $\maxstrategy$ and $\minstrategy$, defined as the only
play conforming to $\maxstrategy$ and $\minstrategy$ and starting
in~$(\loc_0,\val_0)$.

The objective of \MinPl is to reach a target configuration, while minimising
the accumulated weight up to the target. Hence, we associate to every
finite play
$\play=(\loc_0,\val_0)\xrightarrow{d_0,\trans_0}(\loc_1,\val_1)
\xrightarrow{d_1,\trans_1}\cdots
(\loc_k,\val_k)$ its cumulated
weight, taking into account both discrete and continuous costs:
$\weightC(\play)=\sum_{i=0}^{k-1} \weight(\loc_i)\times d_i +
\weight(\trans_i)$.
Then, the weight of a play $\play$, denoted by
$\weight_\game(\play)$, is defined by $+\infty$ if
\play is infinite (does not reach $\LocsT$), and
$\weightC(\play)+\weightT(\loc_T,\val)$ if it ends in
$(\loc_T,\val)$ with $\loc_T\in \LocsT$.
Then, for all locations $\loc$ and
valuation~$\val$, we let $\Val_\game(\loc,\val)$ be the value of
$\game$ in $(\loc,\val)$, defined as
$\Val_\game((\loc,\val)) = \inf_{\minstrategy}\sup_{\maxstrategy}
\weight_\game(\outcomes((\loc,\val),\maxstrategy,\minstrategy))$,
where the order of the infimum and supremum does not matter, since
\WTG{s} are known to be determined\footnote{The determinacy result is
  stated in \cite{BGH+15} for \WTG (called priced timed games) with
  one clock, but the proof does not use the assumption on the number
  of clocks.}.
We say that a strategy $\minstrategy^\star$ of $\MinPl$ is $\varepsilon$-optimal
if, for all $(\loc,\val)$, and all strategies $\maxstrategy$ of $\MaxPl$,
$\weight_\game(\outcomes((\loc,\val),\maxstrategy,\minstrategy^\star))\leq
\Val_\game(\loc,\val)+\varepsilon$. It is said optimal if this holds for
$\varepsilon=0$. A symmetric
definition holds for optimal strategies of $\MaxPl$.
If the game is clear from the context, we may drop the index $\game$
from all previous notations.

As usual in related work \cite{ABM04,BCFL04,BJM15}, we assume that the
input \WTG{s} have guards where all constants are integers,
and all clocks are \emph{bounded}, i.e.~there is a constant
$\clockbound\in\N$ such that every transition of the \WTG is equipped
with a guard $g$ such that $\val\models g$ implies
$\val(x)\leq \clockbound$ for all clocks $x\in \Clocks$. We denote by
$\wmax^\Locs$ (resp.~$\wmax^\Trans$, $\wmax^e$) the maximal weight in
absolute values of locations (resp.~of transitions, edges) of $\game$,
i.e.~$\wmax^\Locs = \max_{\loc\in \Locs} |\weight(\loc)|$
(resp.~$\wmax^\Trans = \max_{\trans \in \Trans} |\weight(\trans)|$,
$\wmax^e = M\wmax^\Locs + \wmax^\Trans$). We also assume that the
output weight functions are piecewise linear with a finite number of
pieces and are continuous on each region. Notice that the zero output
weight function satisfies this property. Moreover, the computations we
will perform in the following maintain this property as an invariant,
and use it to prove their correctness.

\paragraph{Region and corner abstractions}
The region automaton, or region game, $\Nrgame{N}$ (abbreviated as
$\rgame$ when $N=1$) of a game
$\game= \langle\Locs=\LocsMin\uplus\LocsMax, \Trans,\weight,
\LocsT,\weightT\rangle$ is the \WTG with locations
$\RStates = \Locs\times \Nregions N \Clocks \clockbound$ and all
transitions $((\loc,r),\guard'',\reset,(\loc',r'))$ with
$(\loc,\guard,\reset,\loc')\in \Trans$ such that the model of guard $\guard''$
(i.e.~all valuations $\val$ such that $\val\models \guard''$) is a
$1/N$-region $r''$, time successor of $r$ such that $r''$ satisfies
the guard $\guard$, and $r'=r''[\reset:=0]$.  Distribution of locations to
players, final locations and weights are taken according to $\game$. We
call \emph{path} a finite or infinite sequence of transitions in this
automaton, and we denote by $\rpath$ the paths. A play~$\play$ in
$\game$ is projected on a path $\rpath$ in $\Nrgame N$, by replacing
every edge $(\loc,\val)\xrightarrow{d,\trans=(\loc,\guard,\reset,\loc')}(\loc',\val')$ by
the transition $((\loc,r),\guard,\reset,(\loc',r'))$, where $r$ (resp. $r'$) is the $1/N$-region
containing $\val$ (resp. $\val'$): we say that
$\play$ \emph{follows} the path $\rpath$. It is important to notice
that, even if $\rpath$ is a \emph{cycle} (i.e.~starts and ends in the same
location of the region game), there may exist plays following it in
$\game$ that are not cycles, due to the fact that regions are sets of
valuations.  By projecting away the region information of
$\Nrgame{N}$, we simply obtain:
\begin{lemma}\label{lm:region-game}
  For all $\loc\in\Locs$, $1/N$-regions $r$, and
  $\val\in r$,
  $\Val_\game(\loc,\val)=\Val_{\Nrgame N}((\loc,r),\val)$.
\end{lemma}

On top of regions, we will need the corner-point abstraction
techniques introduced in~\cite{BouBri08a}. A valuation~\corv is said
to be a corner of a $1/N$-region~$r$, if it belongs to the
topological closure~$\overline{r}$ and has coordinates multiple of
$1/N$ ($\corv\in(1/N)\N^X$). We call corner state a
triple~$(\loc,r,\corv)$ that contains information about a location
$(\loc,r)$ of the region-game $\Nrgame N$, and a corner $\corv$ of the
$1/N$-region $r$.
Every region has at most~$|\Clocks|+1$ corners.  We now define the
corner-point abstraction~$\Ncgame N$ of a \WTG~$\game$ as the \WTG
obtained as a refinement of $\Nrgame N$ where guards on transitions
are enforced to stay on one of the corners of the current
$1/N$-region: the locations of $\Ncgame N$ are all corner states
of~$\Nrgame N$, associated to each player accordingly, and transitions
are all $((\loc,r,\corv), \guard'', \reset, (\loc',r',\corv'))$ such
that there exists $\rtrans=((\loc,r),\guard,\reset,(\loc',r'))$ a
transition of $\Nrgame N$ such that the model of guard $g''$ is a
corner $\corv''$ satisfying the guard $\overline g$ (recall that
$\overline g$ is the closed version of $g$),
$\corv'=\corv''[\reset:=0]$, and there exist two valuations
$\val\in r$, $\val'\in r'$ such that
$((\loc,r),\val)\xrightarrow{d',\rtrans}((\loc',r'),\val')$ for some
$d'\in \Rplus$ (the latter condition ensures that the transition
between corners is not spurious). Because of this closure operation,
we must also define properly the final weight function: we simply
define it over the only valuation~$v$ reachable in location
$(\loc,r,v)$ (with $\loc\in \LocsT$) by
$\weightT((\loc,r,v),v)=\lim_{\val\to v, \val\in
  r}\weightT(\loc,\val)$ (the limit is well defined since $\weightT$
is piecewise linear with a finite number of pieces on region~$r$).

The \WTG $\Ncgame N$ can be seen as a \emph{weighted game} (with final
weights), i.e.~a \WTG without clocks (which means that there are only
weights on transitions), by removing guards, resets and rates of
locations, and replacing the weights of transitions by the actual
weight of jumping from one corner to another: a transition
$(((\loc,r),\corv), \guard'', \reset, ((\loc',r'),\corv'))$ becomes an
edge from $((\loc,r),\corv)$ to $((\loc',r'),\corv')$ with weight
$d\times \weight(\loc) + \weight(\rtrans)$ (for all possible values of
$d$, which requires to allow for multi-edges\footnote{The only case
  where several edges could link two corners using the same
  transition is when all clocks are reset in $\reset$, in which case
  there is a choice for delay~$d$.}). Note that delay $d$ is
necessarily a rational of the form $\alpha/N$ with $\alpha\in \N$,
since it must relate corners of $1/N$-regions. In particular, this
proves that the cumulated weight $\weightC(\play)$ of a finite play
$\play$ in $\Ncgame N$ is indeed a rational number with
denominator~$N$.

We will call \emph{corner play} a play $\play$ in the corner-point
abstraction $\Ncgame N$: it can also be interpreted as a timed
execution in $\game$ where all guards are closed (as explained in the
definition above). It straightforwardly projects on a finite path
$\rpath$ in the region game $\Nrgame N$: in this case, we say again
that $\play$ follows $\rpath$. Figure~\ref{fig:plays} depicts a play,
its projected path in the region game and one of its associated corner
plays.

\begin{figure}[htbp]
  \centering
  \begin{tikzpicture}[>=latex]
    \path[draw,thick,color=red,fill=red!20!white] (0,0) -- (1,0) --
    (1,1) -- (0,0);

    \path[draw,thick,color=red] (4,0) -- (4,1);

    \path[draw,thick,color=red,fill=red!20!white] (7,0) -- (8,1) --
    (7,1) -- (7,0);

    \path[draw,thick,color=red] (11,0.3) -- (12,0.3);

    \node[draw,circle,inner
    sep=1pt,color=blue!70!black,fill=blue!40!white] (1) at (0.4,0.2) {};
    \node[draw,circle,inner
    sep=1pt,color=blue!70!black,fill=blue!40!white] (2) at (4,0.7) {};
    \node[draw,circle,inner
    sep=1pt,color=blue!70!black,fill=blue!40!white] (3) at (7.3,0.9) {};
     \node[draw,circle,inner
    sep=1pt,color=blue!70!black,fill=blue!40!white] (4) at (11.3,0.3) {};

    \path[thick,color=blue!70!black,->] (1) edge[bend left] (2);
    \path[thick,color=blue!70!black,->] (2) edge[bend left=18] (3);
    \path[thick,color=blue!70!black,->] (3) edge[bend right=40] (4);

    \node[draw,circle,inner
    sep=1pt,color=green!60!black,fill=green!20!white] (1) at (1,0) {};
    \node[draw,circle,inner
    sep=1pt,color=green!60!black,fill=green!20!white] (2) at (4,0) {};
    \node[draw,circle,inner
    sep=1pt,color=green!60!black,fill=green!20!white] (3) at (8,1) {};
    \node[draw,circle,inner
    sep=1pt,color=green!60!black,fill=green!20!white] (4) at (12,0.3) {};

    \path[thick,color=green!60!black,->] (1) edge[bend right] (2);
    \path[thick,color=green!60!black,->] (2) edge[bend right=49] (3);
    \path[thick,color=green!60!black,->] (3) edge[bend left=18] (4);

    \node[color=red] at (0,0.8) {$(\ell_0,r_0)$};
    \node[color=red] at (4,1.4) {$(\ell_1,r_1)$};
    \node[color=red] at (7.5,1.4) {$(\ell_2,r_2)$};
    \node[color=red] at (12.7,0.3) {$(\ell_3,r_3)$};
    \node[color=blue!70!black] at (2.5,1.2) {$\rho$};
    \node[color=green!60!black] at (2.5,-.2) {$\rho'$};

    \path[draw,thick,color=red,->] (1.3,0.3) -- node[above]{$g_0,Y_0$} (3.7,0.3);
    \path[draw,thick,color=red,->] (4.3,0.3) -- node[above]{$g_1,Y_1$} (6.7,0.3);
    \path[draw,thick,color=red,->] (8.3,0.3) -- node[sloped,anchor=south]{$g_2,Y_2$} (10.7,0.3);

  \end{tikzpicture}
  \caption{A play $\rho$ (in blue), its projected path $\pi$ in the
    region game (in red), and one of its associated corner plays
    $\rho'$ (in green)}
  \label{fig:plays}
\end{figure}
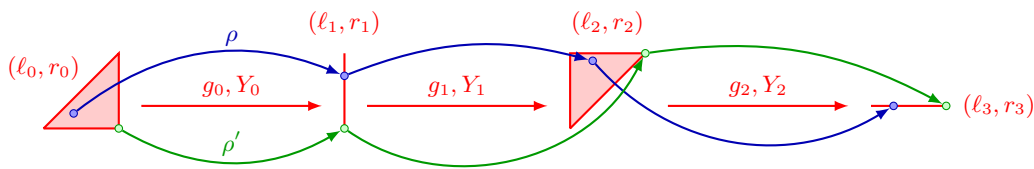

\noindent Corner plays allow one to obtain faithful information on the plays
that follow the same path:

\begin{lemma}\label{lm:cornerabstract}
  If~\rpath is a finite path in~$\Nrgame N$, the
  set~$\{\weightC(\play) \mid \play \text{ finite play following }
  \rpath \}$ is an interval bounded by the minimum and the maximum
  values of the
  set~$\{\weightC(\play) \mid \play \text{ finite corner play
  }\allowbreak \text{of } \Ncgame N \text{ following } \rpath \}$.
\end{lemma}

\paragraph{Value iteration} We will rely on the value iteration algorithm
described in \cite{ABM04} for a \WTG~\game.

If $V$ represents a value function---i.e.~a mapping from configurations of
$\Locs\times\Rplus^\Clocks$ to a value in $\Rbar$---we denote by
$V_{\loc,\val}$ the image $V(\loc,\val)$, for better readability, and
by $V_\loc$ the function mapping each valuation $\val$ to
$V_{\loc,\val}$. One step of the game is summarised in the following
operator $\mathcal{F}$ mapping each value function $V$ to a value
function $V'=\mathcal{F}(V)$ defined by $V'_{\loc,\val}=\weightT(\loc,\val)$ if
$\loc\in\LocsT$, and otherwise
\begin{equation}V'_{\loc,\val}=
\begin{cases}
  \sup_{(\loc,\val)\xrightarrow{d,\trans}(\loc',\val')}
   \big[d\times\weight(\loc)+\weight(\trans)+V_{\loc',\val'}\big]
   &
  \text{if }\loc\in\LocsMax\\
  \inf_{(\loc,\val)\xrightarrow{d,\trans}(\loc',\val')}
   \big[d\times\weight(\loc)+\weight(\trans)+V_{\loc',\val'} \big] &
  \text{if }\loc\in\LocsMin
\end{cases}\label{eq:operator}\end{equation}
\noindent where $(\loc,\val)\xrightarrow{d,\trans}(\loc',\val')$
ranges over valid edges in \game. Then, starting from $V^0$
mapping every configuration to $+\infty$, except for the targets mapped
to $\weightT$, we let $V^i= \mathcal{F}(V^{i-1})$ for all $i>0$. The value
function $V^i$ represents the value $\Val^i_\game$, which
is intuitively what \MinPl can guarantee when forced to
reach the target in at most $i$ steps.

More formally, we define $\weight^i_\game(\play)$
the weight of a maximal play $\play$ at horizon $i$, as
$\weight_\game(\play)$ if $\play$ reaches a target state in at most
$i$ steps, and $+\infty$ otherwise. Using this alternative definition
of the weight of a play, we can obtain a new game value
$\Val^i_\game(\loc,\val)=\inf_{\minstrategy}\sup_{\maxstrategy}
\weight^i_\game(\outcomes((\loc,\val),\maxstrategy,\minstrategy))$.
Then, if $\game$ is a tree of depth $d$,
$V^i {=} \Val_\game$ if $i\ge d$.

The mappings~$V^0_\loc$ are piecewise linear for all
$\loc$, and $\mathcal{F}$ preserves piecewise linearity over regions,
so all iterates $V^i_\loc$ are piecewise linear with a finite number
of pieces. In \cite{ABM04}, it is proved that $V^i_\loc$ has a number
of pieces (and can be computed within a complexity) exponential in~$i$
and in the size of \game when $\weightT=0$. This result can be
extended to handle negative weights in \game and output weights $\weightT\neq0$.

\section{Results}

We consider the \emph{value problem} that asks, given a \WTG $\game$,
a location $\loc$ and a threshold
$\alpha\in \Z\cup\{-\infty,+\infty\}$, to decide whether
$\Val_\game(\loc,\valnull)\leq \alpha$. In the context of timed games,
optimal strategies may not exist. We generally focus on finding
$\varepsilon$-optimal strategies, that guarantee the optimal value, up
to a small error $\varepsilon$.  Moreover, when the value problem is
undecidable, we also consider the \emph{approximation problem} that
consists, given a precision $\varepsilon\in \Q_{>0}$, in
computing
an $\varepsilon$-approximation of $\Val_\game(\loc,\valnull)$.

In the one-player case, computing the optimal
value and an $\varepsilon$-optimal strategy for weighted timed
automata is known to be $\PSPACE$-complete \cite{BouBri07}. In the
two-player case, the value problem of \WTG{s} (also called priced
timed games in the literature) is undecidable with 3 clocks
\cite{BBR05,BJM15}, or even 2 clocks in the presence of negative
weights \cite{BGNK+14} (for the existence problem asking if a strategy
of player \MinPl can guarantee a given threshold). To obtain
decidability, one possibility is to limit the number of clocks to 1:
then, there is an exponential-time algorithm to compute the value as
well as $\varepsilon$-optimal strategies in the presence of
non-negative weights only~\cite{BBM06,Rut11,DueIbs13}, whereas the
problem is only known to be $\P$-hard. A similar result can be lifted
to arbitrary weights, under restrictions on the resets of the clock in
cycles~\cite{BGH+15}.

The other possibility to obtain a decidability
result~\cite{BCFL04,BMR17a} is to enforce a semantical property of
divergence (originally called strictly non-Zeno cost): it asks that
every play following a cycle in the region automaton has weight far
from $0$. It allows the authors to prove that playing for only a
bounded number of steps is equivalent to the original game, which
boils down to the problem of computing the value of a tree-shaped
weighted timed game \game using the value iteration algorithm.

Other objectives, not directly related to optimal reachability, have
been considered in~\cite{BCR14} for weighted timed games, like
mean-payoff and parity objectives. In this work, the authors manage to
solve these problems for the so-called class of $\delta$-robust \WTG{s}
that they introduce. This class includes the class we consider, but is
decidable in 2-$\EXPSPACE$.

In \cite{BMR17a}, we generalised the strictly non-Zeno cost property
of \cite{BCFL04,BMR17a} to weighted timed games with both positive and
negative weights: we called them divergent weighted timed games. This
article relaxes the divergence property, to introduce almost-divergent
weighted timed games. We first define formally these classes of
games. A cycle~\rpath of~\rgame is said to be a positive cycle
(resp.~a 0-cycle, or a negative cycle) if every finite play~\play
following~\rpath satisfies $\weightC(\play)\geq 1$
(resp.~$\weightC(\play)=0$, or $\weightC(\play)\leq -1$). A strongly
connected component (SCC)~$S$ of~\rgame is said to be positive
(resp.~negative) if every cycle~$\rpath\in S$ is positive
(resp.~negative).  An SCC~$S$ of~\rgame is said to be non-negative
(resp.~non-positive) if every play~\play following a cycle in~$S$
satisfies either $\weightC(\play)\geq 1$ or $\weightC(\play)=0$
(resp.~either $\weightC(\play)\leq -1$ or $\weightC(\play)=0$).

\begin{definition}\label{def:divergent}
  A \WTG~\game is divergent if every SCC of~\rgame is either positive or
  negative. As a generalisation, a \WTG~\game is
  almost-divergent when every SCC of~\rgame is either non-negative or
  non-positive.
\end{definition}

In~\cite{BMR17a}, we showed that we can decide in $2\textrm{-}\EXP$
the value problem for divergent \WTG{s}. Unfortunately, it is shown in
\cite{BJM15} that this problem is undecidable for almost-divergent
\WTG{s} (already with non-negative weights only, where
almost-divergent \WTG{s} are called \emph{simple}). They propose a
solution to the approximation problem, again with non-negative weights
only. Our first result is the following extension of their result:
\begin{theorem}\label{thm:almost-div}
  Given an almost-divergent \WTG \game, a location $\loc$ and $\varepsilon \in \Q_{>0}$,
  we can compute an $\varepsilon$-approximation of $\Val_\game(\loc,\valnull)$ in %
   time doubly-exponential in the size of \game
  and polynomial in $1/\varepsilon$.
  Moreover, deciding if a \WTG is almost-divergent is
  \PSPACE-complete.
\end{theorem}

\label{page:SCCs}
To obtain this result, we follow an approximation schema that
we now outline. First, we will always reason on the region game
$\rgame$ of the almost-divergent \WTG $\game$. The goal is to compute
an $\varepsilon$-approximation of $\Val_{\rgame}(s_0,\valnull)$ for
some state $s_0=(\loc_0,r_0)$, with $r_0$ the region where every clock
value is 0.
As already recalled, techniques of \cite{ABM04} allow one to compute
the (exact) values of a \WTG played on a finite tree, using operator
$\mathcal F$. The idea is thus to decompose as much as possible the
game $\rgame$ in a \WTG over a tree. First, we decompose the region
game into SCCs (left of Figure~\ref{fig:schema}).

During the approximation process, we must think about the final weight
functions as the previously computed approximations of the values of
SCCs below the current one.  We will
keep as an invariant that final weight functions are piecewise linear
functions with a finite number of pieces, and are continuous
 on each region.

 For an SCC of $\rgame$ and an initial state $s_0$ of $\rgame$
 provided by the SCC decomposition, we show that the game on the SCC
 is equivalent to a game on a tree built from a semi-unfolding (see
 middle of Figure~\ref{fig:schema}) of $\rgame$ from $s_0$ of finite
 depth, with certain nodes of the tree being \emph{kernels}.  These
 kernels are some parts of \rgame %
 that contain all cycles of weight 0. The semi-unfolding is stopped
 either when reaching a final location, or when some location (or
 kernel) has been visited for a certain fixed number of times: such
 locations deep enough are called \emph{stop leaves}.

\begin{figure}[tbp]
  \centering
  \scalebox{1}{\begin{tikzpicture}[>=latex]
    \begin{scope}[every node/.style={draw,shape=ellipse,minimum
        height=5mm, minimum width=8mm},level/.style={sibling distance=2cm/#1},level distance=8mm,->]
      \node (c) {}
      child { node (a) {}
        child { node (d) {} }
        child { node {} } }
      child { node {}
        child { node (b) {} }
        child { node {} } };
    \end{scope}
    \path[draw,gray,dashed] (1.25,-.5) to[bend right=10] (3.7,.9);
    \path[draw,gray,dashed] (1.25,-1.1) to[bend left=35] (2.6,-2.5);
    \begin{scope}[xshift=4.5cm,yshift=8mm,every
      node/.style={draw,shape=circle,minimum size=5mm,inner sep=0.3mm},level/.style={sibling distance=2cm/#1},level distance=8mm,->]
      \node {$s_0$}
      child { node {$s$}
        child { node[regular polygon,regular polygon sides=6,inner sep=0.1mm] {\makebox[1em][c]{\smaller$\Kernel_{s'}$}}
          child { node {$s$}
            child { node[double,regular polygon,regular polygon sides=3,inner
              sep=0.1mm] (stop) {$s$} }
          }
        }
        child { node {}
          child { node[double,inner sep=0] (out) {$s_f$} }
        }
      }
      child { node {}
        child { node[regular polygon,regular polygon sides=6,inner
          sep=0.1mm] {\makebox[1em][c]{\smaller$\Kernel_{s''}$}}
          child { node[double,inner sep=0] {$s_f$} }
          child { node [double,regular polygon,regular polygon sides=3,inner
              sep=0.1mm] {} }
        }
      };
      \node[node distance=3mm,below of=stop,draw=none] () {\smaller[2]{stop leaf}};
      \node[node distance=5mm,below of=out,draw=none] () {\smaller[2]{$\weightT(s_f)$}};
    \end{scope}
    \path[draw,gray,dashed] (5.8,-.5) to[bend right=10] (8.2,.9);
    \path[draw,gray,dashed] (5.8,-1.1) to[bend left=35] (6.5,-2.5);
    \begin{scope}[xshift=9cm,yshift=6mm,node distance=2cm,scale=0.75, every node/.style={transform shape}]
      \node[draw,circle,minimum size=7mm] (0) {
        \begin{tikzpicture}[scale=0.4, every node/.style={transform shape}]
          \path[draw] (0,0) node[fill=black,minimum size=1mm] {}-- (1,1) -- (1,0) -- cycle;
        \end{tikzpicture}
      };
      \node[draw,circle,below of=0,xshift=-1.5cm,minimum size=7mm] (1) {
         \begin{tikzpicture}[scale=0.4, every node/.style={transform shape}]
          \path[draw] (0,0) node[fill=black,minimum size=1mm] {}-- (1,1) -- (1,0) -- cycle;
        \end{tikzpicture}
      };
      \node[draw,circle,below of=0,xshift=2cm,minimum size=7mm] (2) {
        \begin{tikzpicture}[scale=0.4, every node/.style={transform shape}]
          \path[draw] (0,0)-- (1,1)  node[fill=black,minimum size=1mm] {} -- (1,0) -- cycle;
        \end{tikzpicture}
      };
      \node[draw,rectangle,below of=1,xshift=-1.3cm,minimum size=7mm] (3) {
        \begin{tikzpicture}[scale=0.4, every node/.style={transform shape}]
          \path[draw] (0,0) node[circle,fill=black,minimum size=1mm] {}-- (1,0);
        \end{tikzpicture}
      };
      \node[draw,rectangle,below of=1,xshift=7mm,minimum size=7mm] (4) {
         \begin{tikzpicture}[scale=0.4, every node/.style={transform shape}]
          \path[draw] (0,0) node[circle,fill=black,minimum size=1mm] {}-- (1,1) -- (1,0) -- cycle;
        \end{tikzpicture}
      };
      \node[draw,rectangle,below of=2,xshift=-1cm,minimum size=7mm] (5) {
        \begin{tikzpicture}[scale=0.4, every node/.style={transform shape}]
          \path[draw] (0,0) node[circle,fill=black,minimum size=1mm] {};
        \end{tikzpicture}      };
      \node[draw,circle,below of=2,xshift=1cm,minimum size=7mm] (6) {
        \begin{tikzpicture}[scale=0.4, every node/.style={transform shape}]
          \path[draw] (0,0) node[fill=black,minimum size=1mm] {}-- (1,0) -- (0,-1) -- cycle;
        \end{tikzpicture}
      };

      \path[->] (0) edge node[above left]{0} (1)
      (0) edge node[above right]{$1$} (2)
      (1) edge node[above left]{$-3$} (3)
      (1) edge[bend left=10] node[above right]{$-1$} (4)
      (2) edge node[above left]{$2$} (5)
      (2) edge node[above right]{$1$} (6)
      (3) edge node[below]{$2$} (4)
      (3) edge[bend left=40] node[left]{$3$} (0)
      (4) edge[bend left=10] node[below left]{$1$} (1)
      (4) edge node[below]{$4$} (5)
      (5) edge[bend left=5] node[left]{$-3$} (0)
      (6) edge node[below]{$1$} (5)
      (6) edge[bend left] node[below]{$-3$} (4);
    \end{scope}
  \end{tikzpicture}}
\caption{Static approximation schema: SCC decomposition of \rgame,
  semi-unfolding of an SCC, corner-point abstraction for the kernels}
  \label{fig:schema}
\end{figure}
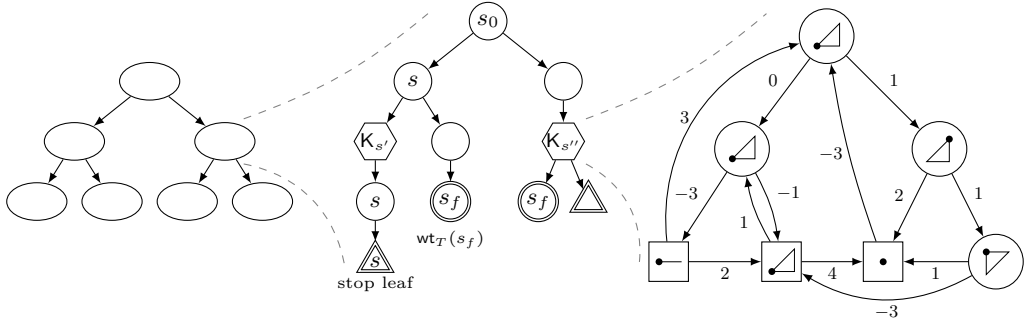

Our second result is a more symbolic approximation schema based on the
value iteration only. It is more symbolic in the sense that it does
not require the SCC decomposition, the computation of kernels nor the
semi-unfolding of the game in a tree.

\begin{theorem}\label{thm:symbolic}
Let $\game$ be an almost-divergent \WTG
such that $\Val_\game>-\infty$ for all configurations.
Then the sequence
$(\Val^k_\game)_{k\geq0}$ converges towards $\Val_\game$ and for
every $\varepsilon\in\Q_{>0}$, we can compute an integer $P$ such that
$\Val^P_\game$ is an $\varepsilon$-approximation of
$\Val_\game$
for all configurations. %
\end{theorem}

\begin{remark}
  In a weighted-timed game, it is easy to detect the set of states
  with value $+\infty$: these are all the states from which \MinPl
  cannot ensure reachability of a target location~$\loc\in\LocsT$ with
  $\weightT(\loc)<+\infty$. It can therefore be computed by an
  attractor computation, and is indeed a property constant on each
  region. In particular, removing those states from~\rgame does not
  affect the value of any other state and can be done in complexity linear
  in $|\rgame|$.
  We will therefore assume that the considered \WTG
  have no configurations with value $+\infty$.%
\end{remark}

\section{Kernels of an almost-divergent \WTG}\label{sec:kernels-in-wtg}

The approximation procedure described before uses the so-called
\emph{kernels} in order to group together all cycles of weight 0. We
study those kernels and give a characterisation allowing computability.
Contrary to the non-negative case,
the situation is more complex in our arbitrary case, since weights of
both locations and transitions may differ from $0$ in the
kernel. Moreover, it is not trivial (and may not be true in a
non almost-divergent \WTG) to know whether it is sufficient to
consider only simple cycles, i.e.~cycles without repetitions.

To answer these questions, let us first analyse the cycles of \rgame
that we will encounter.  Since we are in an almost-divergent game, by
Lemma~\ref{lm:cornerabstract}, all cycles
$\rpath=\rtrans_1\cdots\rtrans_n$ of~\rgame (with
$\rtrans_1,\ldots,\rtrans_n$ transitions of~$\rgame$) are either
0-cycles, positive cycles or negative cycles. Additionally, in an
SCC~$S$ of~\rgame, we cannot find both positive and negative cycles by definition.
Moreover, we can classify a cycle by looking only at the corner plays
following it.
\begin{lemma}\label{lm:exist0}
  A cycle~\rpath is a 0-cycle iff there exists a corner play~\play
  following~\rpath with $\weightC(\play){=}0$.
\end{lemma}
\begin{proof}
  If \rpath is a 0-cycle, every such corner play~\play will have
  weight~$0$, by Lemma~\ref{lm:cornerabstract}.  Reciprocally, if such
  a corner play exists, all corner plays following~\rpath have
  weight~$0$: otherwise the set
  $\{\weightC(\play) \mid \play \text{ play following } \rpath \}$
  would have non-empty intersection with the set
  $(-1,1)\setminus\{0\}$ which would contradict the almost-divergence.
\end{proof}

An important result is that 0-cycles are stable by rotation.  This is
not trivial because plays following a cycle can start and end in
different valuations, therefore changing the starting state of the
cycle could a priori change the plays that follow it and their
weights.
\begin{lemma}\label{lm:rotatcycle}
  Let~\rpath and $\rpath'$ be paths of~\rgame. Then, $\rpath\rpath'$
  is a 0-cycle iff $\rpath'\rpath$ is a 0-cycle.
\end{lemma}
\begin{proof}
  Since~$\rpath_1=\rpath\rpath'$ is a cycle,
  $\first(\rpath)=\last(\rpath')$ and $\first(\rpath')=\last(\rpath)$,
  so~$\rpath_2=\rpath'\rpath$ is correctly defined.

  First, since there are finitely many corners, by constructing a long
  enough play following an iterate of $\rpath'\rpath$, we can obtain a
  corner play that starts and ends in the same corner.  Formally, we
  define two sequences of region
  corners~$(\corv_i\in\first(\rpath))_i$
  and~$(\corv'_i\in\first(\rpath'))_i$.  We start by choosing
  any~$\corv_0\in\first(\rpath)$.  Let~$\corv'_0$ be a corner
  of~$\first(\rpath')$ such that~$\corv'_0$ is accessible
  from~$\corv_0$ by following~\rpath.  For every~$i>0$, let~$\corv_i$
  be a corner of~$\first(\rpath)$ such that~$\corv_i$ is accessible
  from~$\corv'_{i-1}$ by following~$\rpath'$, and let~$\corv'_i$ be a
  corner of~$\first(\rpath')$ such that~$\corv'_i$ is accessible
  from~$\corv_i$ by following~$\rpath$.  We stop the construction at
  the first~$l$ such that there exists~$k<l$ with~$\corv_k=\corv_l$.
  Additionally, we let~$\corv'_l=\corv'_k$
  and~$\corv_{l+1}=\corv_{k+1}$.  This process is bounded
  since~$\first(\rpath)$ has at most~$|\Clocks|+1$ corners.

  For every~$0\leq i\leq l$, let~$w_i$ be the weight of a play
  $\play_i$ from~$\corv_i$ to~$\corv'_i$ along~\rpath, and let~$w'_i$
  be the weight of a play $\play'_i$ from~$\corv'_i$ to~$\corv_{i+1}$
  along~$\rpath'$. The concatenation of the two plays has weight
  $w_i+w'_i=0$, since it follows the 0-cycle $\rpath_1$. Therefore,
  all corner plays from $\corv_i$ to $\corv_i'$ following~\rpath have
  the same weight $w_i$, and the same applies for $w'_i$.  For
  every~$0\leq i< l$, the concatenation of $\play'_i$ and $\play_{i+1}$
  is a play from $\corv'_i$ to $\corv_{i+1}$, of
  weight~$w'_i+w_{i+1}=-w_i+w_{i+1}$,
  following~$\rpath_2$. Since~$\rpath_2$ is a cycle, and the game is
  almost-divergent, all possible values of $w_{i+1}-w_i$ have the same
  sign.

  Finally, we can construct a corner play from $\corv'_k$ to
  $\corv'_l$ by concatenating the plays
  $\play'_k, \play_{k+1}, \allowbreak \play'_{k+1}, \play_{k+2}, \ldots,
  \play'_{l-1}, \play_{l}$. That play has weight
  $\sum_{i=k}^{l-1} (w_{i+1}-w_i)=w_l-w_k=0$. This implies that the
  terms $w_{i+1}-w_i$, of constant sign, are all equal to $0$. As a
  consequence, the concatenation of $\play'_k$ and $\play_{k+1}$ is a
  corner play following $\rpath_2$ of weight $0$. By
  Lemma~\ref{lm:exist0}, we deduce that $\rpath_2$~is a 0-cycle.
\end{proof}

We will now construct the kernel~$\Kernel$ as the subgraph of~\rgame
containing all 0-cycles.  Formally, let \RTransK be the set of
transitions of~\rgame belonging to a \emph{simple} 0-cycle,
and~\RStatesK be the set of states covered by~\RTransK. We define the
kernel~$\Kernel$ of \rgame as the subgraph of~\rgame defined by
\RStatesK and \RTransK.  Transitions in $\RTrans\backslash\RTransK$
with starting state in \RStatesK are called the output transitions
of~\Kernel.  We define it using only simple 0-cycles in order to
ensure its computability. However, we now show that this is of no
harm, since the kernel contains exactly all the 0-cycles, which will
be crucial in the approximation schema we present in
Section~\ref{sec:approx}.

\begin{proposition}\label{prop:kernelprop}
  A cycle of $\rgame$ is entirely in~\Kernel if and only if it is a 0-cycle.
\end{proposition}
\begin{proof}
  We prove that every 0-cycle is in~\Kernel by induction on the length
  of the cycles. The initialisation contains only cycles of length
  $1$, that are in~\Kernel by construction.  If we consider a
  cycle~\rpath of length~$n>1$, it is either simple %
  or it can be rotated and decomposed into~$\rpath'\rpath''$, $\rpath'$ and
  $\rpath''$ being smaller cycles.  Let~\play be a corner play
  following~$\rpath'\rpath''$.  We denote by $\play'$ the prefix
  of~\play following~$\rpath'$ and $\play''$ the suffix
  following~$\rpath''$. It holds that
  $\weightC(\play')=-\weightC(\play'')$, and in an almost-divergent SCC
  this implies $\weightC(\play')=\weightC(\play'')=0$.  Therefore, by
  Lemma~\ref{lm:exist0} both~$\rpath'$ and~$\rpath''$ are 0-cycles,
  and they must be in~\Kernel by induction hypothesis.  Note that this
  reasoning %
  proves that every cycle contained in a longer 0-cycle is also a 0-cycle.

  {\makeatletter
    \let\par\@@par
    \par\parshape0
    \everypar{}\begin{wrapfigure}{r}{3.9cm}
      \vspace{-.3cm}
      \centering
      \scalebox{.8}{
        \begin{tikzpicture}[node distance=5cm,auto,->,>=latex]
          \def \n {5}
          \def \radius {1cm}
          \def \margin {9} %

          \foreach \s in {1,...,\n}
          {
            \node[draw,circle,minimum width=8pt](\s) at ({90+360/\n * (\s - 1)}:\radius) {};
            \draw ({90+360/\n * (\s - 1)+\margin}:\radius)
            arc ({90+360/\n * (\s - 1)+\margin}:{90+360/\n * (\s)-\margin}:\radius)
            node[midway]{$\rtrans_{\s}$};
          }
          \path[looseness=2]
          (1) edge[bend left=100] node[midway]{$\rpath_{\rtrans_5}$} (5)
          (5) edge[bend left=100] node[midway]{$\rpath_{\rtrans_4}$} (4)
          (4) edge[bend left=100] node[midway]{$\rpath_{\rtrans_3}$} (3)
          (3) edge[bend left=100] node[midway]{$\rpath_{\rtrans_2}$} (2)
          (2) edge[bend left=100] node[midway]{$\rpath_{\rtrans_1}$} (1);
        \end{tikzpicture}}
    \end{wrapfigure}
    We now prove that every cycle in \Kernel is a 0-cycle.  By
    construction, every transition $\rtrans\in \RTransK$ is part of a
    simple 0-cycle.  Thus, to every transition~$\rtrans\in \RTransK$,
    we can associate a path~$\rpath_\rtrans$ such that
    $\rtrans\rpath_\rtrans$ is a simple 0-cycle (rotate the simple
    cycle if necessary). We can prove (using both
    Lemmas~\ref{lm:exist0} and \ref{lm:rotatcycle}) the following property by relying
    on another pumping argument on corners:
      If $\rtrans_1\cdots\rtrans_n$ is a path in~\Kernel, then
      $\rtrans_1\rtrans_2\cdots\rtrans_n\rpath_{\rtrans_n}
      \cdots\rpath_{\rtrans_2}\rpath_{\rtrans_1}$ is a 0-cycle
      of~\rgame.
  Now, if $\rpath$ is a cycle of \rgame in~\Kernel, there exists a
  cycle~$\rpath'$ such that $\rpath\rpath'$ is a 0-cycle, therefore
  $\rpath$ is a 0-cycle.\par}%
\end{proof}

\section{Semi-unfolding of almost-divergent
  \WTG{s}}\label{sec:unfolding}

Given an almost-divergent \WTG $\game$, we describe the construction
of its \emph{semi-unfolding} $\tgame$ (as depicted
in~\figurename~\ref{fig:schema}).  This crucially relies on the
absence of states with value~$-\infty$, so we explain how to deal with
them first:
\begin{lemma}\label{lm:-infty}
  In an SCC of $\rgame$, the set of configurations with value
  $-\infty$
  is a union of regions
  computable in time linear in the size of $\rgame$.
\end{lemma}
\begin{proof}[Sketch of proof]
If the SCC is non-negative, the cumulated weight cannot
decrease along a cycle,
  thus,
the only way to obtain value $-\infty$ is to jump in a final state
with final weight $-\infty$. We can therefore compute this set of
states with an attractor for~\MinPl.

If the SCC is non-positive,
  we let~$\RStatesF^\R$ (resp. $\RStatesF^{-\infty}$) be the set
  of target states where~\weightT is
  bounded (resp. has value $-\infty$). We also define $\RTransF^\R$
  (resp.~$\RTransF^{-\infty}$), the set of transitions of \rgame whose
  end state belongs to $\RStatesF^\R$ (resp.~$\RStatesF^{-\infty}$).
  Notice that the kernel cannot contain target states since
  they do not have outgoing transitions.
  We can prove that a configuration has value $-\infty$ iff it belongs to
  a state where player~\MinPl can ensure the LTL formula on
  transitions:
  $( \mathrm G\,\neg\RTransF^\R\wedge \neg \mathrm F \mathrm G
 \, \RTransK )\vee \mathrm F\, \RTransF^{-\infty}$.
  The procedure to detect $-\infty$
  states thus consists of four attractor computations, which can be
  done in time linear in $|\rgame|$.
\end{proof}

We can now assume that no states of $\game$ have value $-\infty$, and
that the output weight function maps all configurations to $\R$.
Since $\weightT$ is piecewise linear with finitely many pieces,
$\weightT$ is bounded. Let $\sup|\weightT|$ denote the bound of
$|\weightT|$, ranging over all target configurations.

We now explain how to build the semi-unfolding $\tgame$. We only build
the semi-unfolding \tgame of an SCC of $\game$ starting from some
state $(\loc_0,r_0)\in S$ of the region game, since it is then easy to
glue all the semi-unfoldings together to get the one of the full game.
Since every configuration has finite value, we can prove that values
of the game are bounded by $|\rgame|\wmax^e+\sup|\weightT|$. As a
consequence, we can find a bound $\gamma$ linear in $|\rgame|$,
$\wmax^e$ and $\sup|\weightT|$ such that a play that visits some state
outside the kernel more than $\gamma$ times has weight strictly above
$|\rgame|\wmax^e+\sup|\weightT|$, hence is useless for the value
computation.
This leads to considering the semi-unfolding $\tgame$ of $\game$
(nodes in the kernel are not unfolded, see Figure~\ref{fig:schema})
such that each node not in the kernel is encountered at most $\gamma$
times along a branch: the end of each branch is called a \emph{stop
  leaf} of the semi-unfolding. In particular, the depth of $\tgame$ is
bounded by $|\rgame| \gamma$, and thus is polynomial in $|\rgame|$,
$\wmax^e$ and $\sup|\weightT|$.  Leaves of the semi-unfolding are thus
of two types: target leaves that are copies of target locations of
$\game$ for which we set the target weight as in $\game$, and stop
leaves for which we set their target weight as being constant to
$+\infty$ if the SCC $\game$ is non-negative, and $-\infty$ if the SCC
is non-positive.

\begin{proposition}\label{prop:semi-unfolding}
  Let $\game$ be an almost-divergent \WTG, and let $(\loc_0,r_0)\in S$
  be some state of the region game. The semi-unfolding $\tgame$ with
  initial state $(\tilde{\loc}_0,r_0)$ (a copy of state
  $(\loc_0,r_0)$) is equivalent to $\game$, i.e.~for all
  $\val_0\in r_0$,
  $\Val_{\game}(\loc_0,\val_0) =
  \Val_{\tgame}((\tilde{\loc}_0,r_0),\val_0)$.
\end{proposition}

\section{Approximation of almost-divergent \WTG{s}}\label{sec:approx}

\paragraph{Approximation of kernels}\label{sec:approx-kernels}
We start by approximating a kernel $\game$
by extending the region-based approximation schema of
\cite{BJM15}. In their setting, all runs in kernels had weight 0, allowing
a simple reduction to a finite weighted game. In our setting, we
have to approximate the timed dynamics of runs, and therefore
resort to the corner-point abstraction (as shown to the right of~\figurename~\ref{fig:schema}).

Since output weight functions are piecewise linear with a finite number
of pieces and continuous on regions, they are \lipconst-Lipschitz-continuous\footnote{The
  function $\weightT$ is said to be \lipconst-Lipschitz-continuous when
  $|\weightT(s,\val)-\weightT(s,\val')|\leq \lipconst \|\val-\val'\|_\infty$
  for all valuations $\val,\val'$, where
  $\|v\|_\infty=\max_{x\in \Clocks} |v(x)|$ is the $\infty$-norm of
  vector $v\in\R^\Clocks$.
  The function $\weightT$ is said to be
  Lipschitz-continuous if it is \lipconst-Lipschitz-continuous, for some
  \lipconst.}, for a
given constant $\lipconst\geq 0$.
We let $\myconst = \wmax^\Locs\,|\Locs|
|\regions\Clocks\clockbound| + \lipconst$.

Let $N$ be an integer. Consider the game $\Ncgame{N}$ described in the
preliminary section, with locations of the form $(\loc,r,v)$ with $v$ a
corner of the $1/N$-region $r$. Two plays $\play$ of $\game$ and
$\play'$ of $\Ncgame{N}$ are said to be \emph{$1/N$-close} if
they follow the same path \rpath in $\Nrgame N$.
In particular, at each step the configurations $(\loc,\val)$ in $\play$ and $(\loc',r',v')$
in $\play'$ (with $v'$ a corner of the $1/N$-region $r'$)
satisfy $\loc=\loc'$ and $\val\in r'$, and the transitions taken in both plays
have the same discrete weights. Close plays have \emph{close} weights,
in the following sense:

\begin{lemma}\label{lm:close-plays}
  For all $1/N$-close plays $\play$ of $\game$ and $\play'$ of $\Ncgame{N}$,
  $|\weight_\game(\play)-\weight_{\Ncgame N}(\play')|\leq \myconst/N$.
\end{lemma}

In particular, if we start in configurations $(\loc_0,\val_0)$ of
$\game$, and $((\loc_0,r_0,v_0),v_0)$ of $\Ncgame N$, with
$\val_0\in r_0$, since both players have the ability to stay
$1/N$-close all along the plays, a bisimulation argument permits to
obtain that the values of the two games are also close in
$(\loc_0,\val_0)$ and $((\loc_0,r_0,v_0),v_0)$:

\begin{lemma}\label{lm:bisimulation}
  For all locations $\loc\in \Locs$, $1/N$-regions $r$, $\val\in r$ and
  corners $v$ of $r$,
  $|\Val_\game(\loc,\val)-\Val_{\Ncgame N}((\loc,r,v),v)|\leq \myconst/N$.
\end{lemma}

Using this result, picking $N$ an integer larger than
$\myconst/ \varepsilon$, we
can thus obtain
$|\Val_\game(\loc,\val)-\Val_{\Ncgame N}((\loc,r,v),v)|\leq
\varepsilon$.
Recall that $\Ncgame N$ can be considered as an untimed weighted game
(with reachability objective). Thus we can apply the result
of~\cite{BGHM16}, where it is shown that the optimal values of such
games can be computed in pseudo-polynomial time (i.e.~polynomial time
with weights encoded in unary, instead of binary). We then
define an $\varepsilon$-approximation of $\Val_\game$, named $\Val_N'$,
 on each $1/N$-region by interpolating the values of its $1/N$-corners in $\Ncgame N$
 with a piecewise linear function:
therefore,
we can control the Lipschitz constant of the approximated value for
further use.
\begin{lemma}\label{lm:kernel-approx-regular}
  $\Val_N'$ is an $\varepsilon$-approximation of
  $\Val_\game$, that is piecewise linear with a finite number of
  pieces and
  $2\myconst$-Lipschitz-continuous over regions.
\end{lemma}

\paragraph{Approximation of almost-divergent
  \WTG{s}}\label{sec:approx-AD}
We now explain how to approximate the value of an
almost-divergent \WTG $\game$, thus proving Theorem~\ref{thm:almost-div}.
First, we compute a semi-unfolding $\tgame$ as
described in the previous section.
Then we perform a bottom-up computation of the approximation.
As already recalled, techniques of~\cite{ABM04} allow us to compute
exact values of a tree-shape \WTG.
In consequence, we know how to compute the value of a non-kernel node of $\tgame$,
depending of the values of its children. There is no approximation needed here, so that
if all children are $\varepsilon$-approximation, we can compute an
$\varepsilon$-approximation of the node.
Therefore, the only approximation lies in the kernels, and
we explained before how to compute arbitrarily close
an approximation of a kernel's value.
We crucially rely on the fact that the value function is
1-Lipschitz-continuous\footnote{Indeed, $\inf$ and $\sup$ are
  1-Lipschitz-continuous functions, and with a fixed play \play, the
  mapping $\weightT\to\weightC(\play)+\weightT(\last(\play))$ is
  1-Lipschitz-continuous.}.  This entails that imprecisions will sum
up along the bottom-up computations, as computing an
$\varepsilon$-approximation of the value of a game whose output
weights are $\varepsilon'$-approximations yields an
$(\varepsilon+\varepsilon')$-approximation.  Therefore we compute
approximations with threshold $\varepsilon'=\varepsilon/\alpha$ for
kernels in $\tgame$, where $\alpha$ is the maximal number of kernels
along a branch of $\tgame$: $\alpha$ is smaller than the depth of
$\tgame$, which is bounded by Proposition~\ref{prop:semi-unfolding}.

The subregion granularity considered before for kernel approximation
crucially depends on the Lipschitz constant of output weights.  The
growth of these constants is bounded for kernels in $\tgame$ by
Lemma~\ref{lm:kernel-approx-regular}. For non-kernel nodes of
$\tgame$, using a careful analysis of the algorithm of~\cite{ABM04},
we obtain the following bound:
\begin{lemma}\label{lm:tree-lipschitz}
  If all the output weights of a WTG \game are \lipconst-Lipschitz-continuous over regions
  (and piecewise linear, with finitely many pieces),
  then $\Val^i_\game$ is %
  $\lipconst\lipconst'$-Lipschitz-continuous over regions,
  with $\lipconst'$ polynomial in $\wmax^\Locs$ and $|\Clocks|$
  and exponential in $i$.
\end{lemma}

The overall time complexity of this method is doubly-exponential in
the size of the input game and polynomial in $1/\varepsilon$.

\section{Symbolic approximation algorithm}\label{sec:symbolic}

The previous approximation result suffers from several drawbacks. It
relies on the SCC decomposition of the region automaton. Each of these
SCCs have to be analysed in a sequential way, and their analysis
requires an a priori refinement of the granularity of regions.  This
approach is thus not easily amenable to implementation. We instead
prove in this section that the symbolic approach based on the value
iteration paradigm, i.e.~the computation of iterates of the
operator~$\mathcal F$ recalled in page~\pageref{eq:operator}, is an
approximation schema.  This is stated in~Theorem~\ref{thm:symbolic},
for which we now sketch a proof in this section.

Notice that configurations with value $+\infty$ are stable through
value iteration, and do not affect its other computations.  Since
Theorem~\ref{thm:symbolic} assumes the absence of configurations of
value $-\infty$, we will therefore consider in the following that all
configurations have finite value in \game.

Consider first a game \game that is a kernel. By the results of
Section~\ref{sec:approx-kernels}, there exists an integer~$N$ such
that solving the untimed weighted game $\Ncgame N$ computes an
$\varepsilon/2$-approximation of the value of $1/N$ corners.
Using the results of~\cite{BGHM16} for untimed weighted games, we know
that those values are obtained after a finite number of steps of (the
untimed version of) the value iteration operator.  More precisely, if
one considers a number of iterations
$P = |\Locs| |\Nregions N\Clocks\clockbound| (|\Clocks|+1) (2 (|\Locs|
|\Nregions N\Clocks\clockbound| (|\Clocks|+1) -1) \wmax^e+1)$, then
$\Val^P_{\Ncgame N}((\loc,r,v),v)=\Val_{\Ncgame N}((\loc,r,v),v)$.
From this observation, we deduce the following property of $P$:
\begin{lemma}\label{lm:symbolic-kernel}
  If $\game$ is a kernel with no configurations of infinite value,
  then
  $|\Val_\game(\loc,\val)-\Val^P_\game(\loc,\val)|\leq \varepsilon$
  for all configurations $(\loc,\val)$ of $\game$.
\end{lemma}
\begin{proof}
  We already know that
  $\Val^P_{\Ncgame N}((\loc,r,v),v)=\Val_{\Ncgame N}((\loc,r,v),v)$
  for all configurations $((\loc,r,v),v)$ of $\Ncgame N$. Moreover,
  Section~\ref{sec:approx-kernels} ensures
  $|\Val_\game(\loc,\val)-\Val_{\Ncgame N}((\loc,r,v),v)|\leq
  \varepsilon/2$ whenever $\val$ is in the $1/N$-region
  $r$. Therefore, we only need to prove that
  $|\Val^P_\game(\loc,\val)-\Val^P_{\Ncgame N}((\loc,r,v),v)|\leq
  \varepsilon/2$ to conclude. This is done as for
  Lemma~\ref{lm:bisimulation}, since Lemma~\ref{lm:close-plays} (that
  we need to prove Lemma~\ref{lm:bisimulation}) does not depend on the
  length of the plays $\play$ and $\play'$, and both runs reach the
  target state in the same step, i.e.~both before or after the horizon
  of $P$ steps.
\end{proof}

Once we know that value iteration converges on kernels, we can use the semi-unfolding
of Section~\ref{sec:unfolding} to prove that it also converges on non-negative SCCs
when all values are finite.
\begin{lemma}\label{lm:symbolic-scc-plus}
  If $\game$ is a non-negative SCC with no configurations of infinite
  value, we can compute $P_+$ such that
  $|\Val_\game(\loc,\val)-\Val^{P_+}_\game(\loc,\val)|\leq \varepsilon$
  for all configurations $(\loc,\val)$ of $\game$.
\end{lemma}

The idea is to unfold every kernel of the semi-unfolding game \tgame
according to its bound in Lemma~\ref{lm:symbolic-kernel}. More
precisely, let $\alpha$ be the maximum number of kernels along one of
the branches of $\tgame$. In a bottom-up fashion, we can find for each
kernel \Kernel in \tgame a bound~$P_\Kernel$ such that, for all
configurations $(\loc,\val)$,
$|\Val_\Kernel(\loc,\val)-\Val^{P_\Kernel}_\Kernel(\loc,\val)|\leq
\varepsilon/\alpha$. We thus unfold \Kernel in \tgame with depth up
to~$P_\Kernel$.  After each kernel has been replaced this way, \tgame is no
longer a semi-unfolding, it is instead a (complete) unfolding of
\rgame, of a certain bounded depth~$P_+$. This new bound $P_+$ is
bounded by the former depth of \tgame to which is added $\alpha$ times
the biggest bound $P_\Kernel$ we need for the kernels.  Now, \tgame is a
tree of depth $P_+$ whose value at its root is $\varepsilon$-close to
the value of \game.  Finally, the value computed by $\Val^{P_+}_{\game}$
is bounded between $\Val_\game$ and $\Val_\tgame$, which allows us to
conclude.

The bound $P_\Kernel$ for a kernel \Kernel depends linearly in \lipconst,
the Lipschitz constant of value functions on locations of \tgame
reachable from \Kernel. Once \Kernel has been replaced by its
unfolding of depth $P_\Kernel$, the Lipschitz constant of the value function
at the root of \tgame are thus bounded exponentially in \lipconst.
This means that we ensure a bound for $P_+$ that is at most polynomial
in $1/\varepsilon$, and that is of the order of a tower of $\alpha$
exponentials.

Proving the same property on non-positive SCCs requires more work,
because the semi-unfolding gives output weight $-\infty$ to stop
leaves, %
which doesn't integrate well with value iteration (initialisation at
$+\infty$ on non-target states). However, by unfolding those SCCs
slightly more (at most $|\rgame|$ more steps), we can obtain the
desired property with a similar bound~$P_-$.
\begin{lemma}\label{lm:symbolic-scc-minus}
  If $\game$ is a non-positive SCC with no configurations of infinite
  value, we can compute $P_-$ such that
  $|\Val_\game(\loc,\val)-\Val^{P_-}_\game(\loc,\val)|\leq
  \varepsilon$ for all configurations $(\loc,\val)$ of $\game$.
\end{lemma}

Now, if we are given an almost-divergent game \game and a precision $\varepsilon$,
we can add the bounds for value iteration obtained from each SCC by Lemmas~\ref{lm:symbolic-scc-plus} and~\ref{lm:symbolic-scc-minus},
and obtain a final bound $P$ such that for all $k\geq P$,
$\Val^k_\game$ is an $\varepsilon$-approximation of $\Val_\game$.

\smallskip
\paragraph{Discussion}
Overall, this leads to an upper bound complexity that is polynomial in
$1/\varepsilon$ and of the order of a tower of $n$ exponentials, with
$n$ polynomial in the size of the input \WTG.
However, we argue that this symbolic procedure is more amenable to
implementation than the previous approximation schema. First, it
avoids the three already mentioned drawbacks (SCC decomposition,
sequential analysis of the SCCs, and refinement of the granularity of
regions) of the previous approximation schema. Then, it allows one to
directly launch the value iteration algorithm on the game $\game$, and
we can stop the computation whenever we are satisfied enough by the
approximation computed: however, there are no guarantees whatsoever on
the quality of the approximation before the number of steps $P$ given
above. Finally, this schema allows one to easily obtain an
almost-optimal strategy with respect to the computed value.

If $\game$ is not guaranteed to be free of configurations of value
$-\infty$, then we must first perform the SCC decomposition of \rgame,
and, as \game is almost-divergent, identify and remove regions whose
value is $-\infty$, by Lemma~\ref{lm:-infty}. Then, we can apply the
value iteration algorithm.

As a final remark, notice that our correctness proof strongly relies
on Section~\ref{sec:approx-kernels}, and thus would not hold with the
approximation schema of \cite{BJM15} (which
does not preserve the continuity on regions of the computed value
functions, in turn needed to define output weights on $1/N$-corners).

\section{Conclusion}

We have given an approximation procedure for a large
class of weighted timed games with unbounded number of clocks and
arbitrary integer weights
that can be executed in doubly-exponential time with respect to
the size of the game.
In addition, we proved the correctness of a symbolic approximation schema,
that does not start by splitting exponentially every region, but only
does so when necessary (as dictated by~\cite{ABM04}). We argue that this paves
the way towards an implementation of value approximation for weighted timed games.

Another perspective is to extend this work to the concurrent setting,
where both players play simultaneously and the shortest delay is
selected.  We did not consider this setting in this work because
concurrent \WTG{s} are not determined, and several of our proofs rely
on this property for symmetrical arguments (mainly to lift results of
non-negative SCCs to non-positive ones). Another extension of this work is the
exploration of the effect of almost-divergence in the case of multiple weight
dimensions, and/or with mean-payoff objectives.

\bibliographystyle{plain}

\newpage
\appendix

\section{Proofs of Section~\ref{sec:prelim}}

\begin{proof}[Proof of Lemma~\ref{lm:cornerabstract}]
  The set
  $\{\weightC(\play) \mid \play \text{ finite play following } \rpath
  \}$ is an interval as the image of a convex set by an linear
  function (see \cite[Sec.~3.2]{BouBri07} for an explanation).  The
  good properties of the corner-point abstraction allows us to
  conclude, since for every play~\play following~\rpath, one can find
  a corner play following~\rpath of smaller weight and one of larger
  weight, and for every corner play~\play following~\rpath and
  every~$\varepsilon>0$, one can find a play following~\rpath whose
  weight is at most~$\varepsilon$ away from~$\weightC(\play)$
  \cite{BouBri08a}.
\end{proof}

\subsection{Undecidability of value $-\infty$}\label{app:-infty_undec}

We prove that given a \WTG $\game$ (not necessarily almost-divergent)
and an initial location~$\loc_0$, it is undecidable whether
$\Val_\game(\loc_0,\valnull)=-\infty$. We reduce it to the existence
problem on turn-based \WTG: given a \WTG \game (without output weight
function), an integer threshold $\alpha$ and a starting location
$\loc_0$, does there exist a strategy for \MinPl that can guarantee
reaching the unique target location $\loc_t$ from $\loc_0$ with weight
$<\alpha$. In the non-negative setting, it is proved
in~\cite{BBM06} that the problem is undecidable for the comparison
$\leq \alpha$. In the negative setting, formal proofs are given for
all comparison signs in~\cite{BGNK+14}.

Consider $\game'$ the \WTG built from \game by adding a transition
from $\loc_t$ to $\loc_0$, without guards and resetting all the
clocks, of discrete weight $-\alpha$. We add a new target location
$\loc_t'$, and add transitions of weight $0$ from $\loc_t$ to
$\loc_t'$.  Location $\loc_t$ is then given to \MinPl. Let us prove
that $\Val_{\game'}(\loc_0,\valnull)=-\infty$ if and only if \MinPl
has a strategy to guarantee a weight $<\alpha$ in $\game$.

Assume first $\Val_{\game'}(\loc_0,\valnull)=-\infty$. If
$\Val_\game(\loc_0,\valnull)=-\infty$, we are done. Otherwise, \MinPl
must follow in $\game'$ the new transition from $\loc_t$ to $\loc_0$
to enforce a cycle of negative value, and thus enforce a play from
$(\loc_0,\valnull)$ to $\loc_t$ with weight less than $\alpha$.
Therefore, there exists a strategy for \MinPl in $\game$ that can
guarantee a weight $<\alpha$.

Reciprocally, if there exists a strategy for \MinPl that can guarantee
a weight $<\alpha$, then \MinPl can force a negative cycle play and
$\Val_{\game'}(\loc_0,\valnull)=-\infty$.

\subsection{Decision of the almost-divergence of a \WTG}\label{app:class-decision}

First, we state that a \WTG \game is not almost-divergent if and only
if \rgame contains an SCC with either both a positive play following
one of its cycles and a negative play following one of its cycles, or
a play with weight in $(-1,0)\cup(0,1)$ following one of its cycles.
We will now explain how we can test both of those properties (and thus
if a game is not almost-divergent) in \PSPACE.

A corner play following a cycle of the region game is said to be
simple if it does not visit the same corner twice (but the first and
last corners can be the same).  A simple corner play following a cycle
has length bounded by $|\RStates|\times(|\Clocks|+1)$.  By
Lemma~\ref{lm:cornerabstract}, \rgame contains an SCC with either both
a positive play following one of its cycles and a negative play
following one of its cycles if and only if \rgame contains both a
positive corner play following one of its cycles and a negative corner
play following one of its cycles. We will extend this to simple corner
plays.

\begin{lemma}
  \rgame contains an SCC with either both a positive play following
  one of its cycles and a negative play following one of its cycles if
  and only if \rgame contains an SCC with both a positive simple
  corner play following one of its cycles and a negative simple corner
  play following one of its cycles.
\end{lemma}
\begin{proof}
  All that is left to prove is that, in an SCC of \rgame, if all
  simple corner plays following a cycle have non-negative weight
  (resp.~non-positive weight), then all corner plays following
  a cycle have non-negative weight (resp.~non-positive weight).

  By contradiction, we consider \play, the shortest corner play
  following a cycle \rpath, such that $\weightC(\play)<0$
  (resp.~$\weightC(\play)>0$). Corner play \play cannot be
  simple, so it must contain a simple loop.  That loop is a simple
  corner play following a cycle of \rgame, so it must have
  non-negative weight (resp.~non-positive weight).  This means
  that \play without that loop satisfies $\weightC(\play)<0$
  (resp.~$\weightC(\play)>0$), and therefore was not the
  shortest corner play with the desired property.
\end{proof}

We can test the existence of such simple corner plays in a SCC of
\rgame in \NPSPACE, by guessing them corner after corner and by
keeping the cumulated weight in memory.  The check that both plays are
in the same SCC is a reachability check in a timed automaton, which
can be done in \PSPACE.  We described a similar procedure
in~\cite{BMR17a} where we were testing the existence of a non-negative
corner play and a non-positive one in the same SCC instead of a
negative one and a positive one.

Now, we will assume in this second part that this test failed, so
every SCC of \rgame either satisfies that all plays following a cycle
have non-negative weight or satisfies that they all have non-positive
weight.  We will now explain how to check if \rgame contains a play
with weight in $(-1,0)\cup(0,1)$ following one of its cycles.  Let
$B=(|\RStates|\times(|\Clocks|+1))^2$.

\begin{lemma}
  \rgame contains a play with weight in $(-1,0)\cup(0,1)$ following
  one of its cycles if and only if \rgame contains a cycle \rpath of
  length at most $B$ such that there is a corner play following \rpath
  with weight zero and another one with non-zero weight.
\end{lemma}
\begin{proof}
  By Lemma~\ref{lm:cornerabstract}, \rgame contains a play with weight
  in $(-1,0)\cup(0,1)$ following one of its cycles if and only if that
  cycle satisfies that there is a corner play following it with weight
  zero and another one with non-zero weight.

  We only need to show that if there are no such cycles of length at
  most $B$, then there are no such cycles of any length.  Therefore,
  we assume that no cycle of length less than $B$ allows a play with
  weight in $(-1,0)\cup(0,1)$.  By contradiction, let \rpath be the
  shortest cycle such that there exist two corner plays \play and
  $\play'$ following \rpath, with $\weightC(\play)=0$ and
  $\weightC(\play')\neq 0$. Then $|\rpath|>B$.  Let $\corv_i$ be the
  $i$-th corner of \play, and $\corv'_i$ be the $i$-th corner of
  $\play'$. There are at most $(|\RStates|\times(|\Clocks|+1))^2$
  different pairs $(\corv_i,\corv'_i)$, which implies that there must
  be two indexes, $j$ and $k$, such that
  $(\corv_j,\corv'_j)=(\corv_k,\corv'_k)$ and $j<k$.  The portion of
  \play between indexes $j$ and $k$ follows a cycle, and have opposite
  weight to the play constructed by considering \play and removing the
  loop between indexes $j$ and $k$.  Since the sum of their weight is
  0 and they both follow cycles of \rgame in the same SCC, both of
  those plays have weight 0.  The portion of \rpath between indexes
  $j$ and $k$ is a cycle shorter than \rpath, and it contains a corner
  play of weight 0, therefore all of its corner plays have weight 0,
  and the portion of $\play'$ between indexes $j$ and $k$ has weight 0
  too.  But then the cycle defined by taking \rpath and removing the
  loop between indexes $j$ and $k$ contains a corner play of weight 0
  (derived from \play), and a corner play of weight non-zero (derived
  from $\play'$), and that contradicts \rpath being the shortest cycle
  with that property.
\end{proof}

Once again, we can check the existence of such a cycle of length
bounded by $B$ in \NPSPACE\ by guessing it and its two relevant corner
plays on-the-fly and storing the cumulated weight of each.  This imply
that deciding if a game \game is almost divergent is decidable in
$\coNPSPACE=\NPSPACE=\PSPACE$ (using the theorems of
Immerman-Szelepcs\'enyi~\cite{Imm88,Sze88} and Savitch~\cite{Sav70}).

Let us now show the $\PSPACE$-hardness (indeed the $\co\PSPACE$, which
is identical) by a reduction from the reachability problem in a timed
automaton. We consider a timed automaton
with a starting state and a different target state without outgoing
transitions. We construct from it a weighted timed game by
distributing all states to \MinPl, and equipping all transitions with
weight $0$, and all states with weight $0$. We also add a loop with
weight $1$ on the initial state, one with weight $-1$ on the target state, and a transition from the target
state to the initial state with weight $0$, all three resetting all
clocks and with no guard. Then, the weighted timed game
is not almost-divergent if and only if the target can be reached from the
initial state in the timed automaton.

\section{Proofs of the kernel characterisation (Section~\ref{sec:kernels-in-wtg})}\label{app:kernels}

\begin{lemma}\label{lm:cyclescombi}
  If $\rtrans_1\cdots\rtrans_n$ is a path in~\Kernel, then
  $\rtrans_1\rtrans_2\cdots\rtrans_n\rpath_{\rtrans_n}
  \cdots\rpath_{\rtrans_2}\rpath_{\rtrans_1}$ is a 0-cycle
  of~\rgame.
\end{lemma}
\begin{proof}%
  We prove the property by induction on~$n$. For~$n=1$, the property
  is immediate since $\rtrans_1\rpath_{\rtrans_1}$ is a 0-cycle.
  Consider then $n$ such that the property holds for $n$, and prove it
  for $n+1$. We will exhibit two corner plays following
  $\rtrans_1\cdots\rtrans_{n+1}\rpath_{\rtrans_{n+1}}\cdots\rpath_{\rtrans_1}$
  of opposite weight and conclude with Lemma~\ref{lm:exist0}.

  Let~$\corv_0$ be a corner of~$\last(\rtrans_{n+1})$. Since
  $\rtrans_{n+1}\rpath_{\rtrans_{n+1}}$ is a 0-cycle, there
  exists~$w\in\Z$, a corner play~$\play_{0}$ following~$\rtrans_{n+1}$
  ending in~$\corv_0$ with weight~$w$ and a corner play~$\play'_{0}$
  following~$\rpath_{\rtrans_{n+1}}$ beginning in~$\corv_0$ with
  weight~$-w$. We name~$\corv'_0$ the corner of~$\last(\rtrans_n)$
  where ends $\play'_{0}$. We consider any corner play $\play_1$
  following $\rtrans_{n+1}$ from corner $\corv'_0$. The corner play
  $\play'_0\play_1$ follows the path
  $\rpath_{\rtrans_{n+1}}\rtrans_{n+1}$ that is also a $0$-cycle by Lemma~\ref{lm:rotatcycle},
  therefore $\play_1$ has weight $w$. We denote by $\corv_1$ the
  corner where ends $\play_1$. By iterating this construction, we
  obtain some corner plays $\play_0,\play_1,\play_2,\ldots$ following
  $\rtrans_{n+1}$ and $\play'_0,\play'_1,\play'_2,\ldots$ following
  $\rpath_{\rtrans_{n+1}}$ such that $\play'_{i}$ goes from corner
  $\corv_i$ to $\corv'_i$, and $\play_{i+1}$ from corner $\corv'_i$ to
  $\corv_{i+1}$, for all $i\geq 0$. Moreover, all corner plays
  $\play_i$ have weight $w$ and all corner plays $\play'_i$ have
  weight $-w$. Consider the first index $l$ such that
  $\corv_l=\corv_k$ for some $k<l$, which exists because the number of
  corners is finite.

  We apply the induction to find a corner play following
  $\rtrans_1\cdots\rtrans_{n}\rpath_{\rtrans_{n}}\cdots\rpath_{\rtrans_1}$,
  going through the corner $\corv'_k$ in the middle: more formally,
  there exists $w_\alpha$, a corner play~$\play_\alpha$ following
  $\rtrans_1\cdots\rtrans_n$ ending in~$\corv'_k$ with
  weight~$w_\alpha$ and a corner play~$\play'_\alpha$
  following~$\rpath_{\rtrans_n}\cdots\rpath_{\rtrans_1}$ beginning
  in~$\corv'_k$ with weight~$-w_\alpha$. We apply the induction a
  second time with corner $\corv'_{l-1}$: there exists $w_\beta$, a
  corner play~$\play_\beta$ following $\rtrans_1\cdots\rtrans_n$
  ending in~$\corv'_{l-1}$ with weight~$w_\beta$ and a corner
  play~$\play'_\beta$
  following~$\rpath_{\rtrans_n}\cdots\rpath_{\rtrans_1}$ beginning
  in~$\corv'_{l-1}$ with weight~$-w_\beta$.

  The corner play
  $\play_\alpha\play_{k+1}\play'_{k+1}\play_{k+2}\play'_{k+2}
  \cdots\play'_{l-1}\play'_\beta$, of weight
  $w_\alpha+(w-w)^{l-k}-w_\beta=w_\alpha-w_\beta$, follows the cycle
  $\rtrans_1\cdots\rtrans_n(\rtrans_{n+1}\rpath_{\rtrans_{n+1}})^{l-k}
  \rpath_{\rtrans_n}\cdots\rpath_{\rtrans_1}$. The corner play
  $\play_\beta\play_l\play'_k\play'_\alpha$, of weight
  $w_\beta+w-w-w_\alpha=w_\beta-w_\alpha$, follows the cycle
  $\rtrans_1\cdots\rtrans_n\rtrans_{n+1}
  \rpath_{\rtrans_{n+1}}\rpath_{\rtrans_n}\cdots\rpath_{\rtrans_1}$. Since
  the game is almost-divergent, and those two corner plays are in the
  same SCC, both have weight 0. The second corner play of weight 0
  ensures that the cycle
  $\rtrans_1\cdots\rtrans_{n+1}
  \rpath_{\rtrans_{n+1}}\cdots\rpath_{\rtrans_1}$ is a 0-cycle, by
  Lemma~\ref{lm:exist0}.
\end{proof}

\section{Proofs of the semi-unfolding  (Section~\ref{sec:unfolding})}\label{app:unfolding}

\begin{proof}[Proof of Lemma~\ref{lm:-infty}]
  We detail the case of non-negative SCCs.
  Let us prove that a configurations has value $-\infty$ if and only if it
  belongs to a state where player~\MinPl can ensure the LTL
  formula on transitions:
  $\phi = ( \mathrm G(\neg\RTransF^\R)\wedge \neg \mathrm F \mathrm G
  \RTransK )\vee \mathrm F \RTransF^{-\infty}$.
  Since $\omega$-regular
  games are determined, this is equivalent to saying that a
  configuration has finite value if and only if it belongs to a state
  where~\MaxPl can ensure $\lnot \phi$.

  If \rstate is a state where~\MinPl can ensure $\phi$, he can ensure
  $-\infty$ value from all configurations in~\rstate by either
  reaching~$\RStatesF^{-\infty}$ or avoiding $\RStatesF^\R$ for as
  long as he desires, while not getting stuck in~\Kernel, and thus
  going through an infinite number of negative cycles by
  Proposition~\ref{prop:kernelprop}.  This proves that a state
  where~\MaxPl cannot ensure $\lnot \phi$ contains only valuations of value
  $-\infty$.  Conversely, if \rstate is a state where~\MaxPl can
  ensure
  $\lnot \phi =(\mathrm F \RTransF^\R \vee \mathrm F \mathrm G
  \RTransK) \wedge \mathrm G \neg \RTransF^{-\infty}$, then from~$s$,
  \MaxPl must be able to avoid $\RStatesF^{-\infty}$, and eventually
  enforce either $\RStatesF^\R$ reachability or staying in \Kernel
  forever. In both cases, \MaxPl can ensure a value above $-\infty$.
\end{proof}

\subsection{Semi-unfolding construction}

In order to prove Proposition~\ref{prop:semi-unfolding}, we will construct
the desired semi-unfolding \tgame of a (non-negative or non-positive SCC) \game.

If $(\loc,r)$ is in~\Kernel,
we let~$\Kernel_{\loc,r}$ be the part of~\Kernel accessible
from~$(\loc,r)$ (note that $\Kernel_{\loc,r}$ is an SCC as \Kernel is
a disjoint set of SCCs).  We define the output transitions
of~$\Kernel_{\loc,r}$ as being the output transitions of~\Kernel
accessible from~$(\loc,r)$.  If~$(\loc,r)$ is not in~\Kernel, the
output transitions of~$(\loc,r)$ are the transitions of~\rgame
starting in~$(\loc,r)$.

Formally, we define a tree $T$ whose nodes will either be labelled by
region graph states $(\loc,r)\in\RStates\backslash\RStatesK$ or by
kernels $\Kernel_{\loc,r}$, and whose edges will be labelled by output
transitions in~\rgame. The root of the tree~$T$ is labelled
with~$(\loc_0,r_0)$, or~$\Kernel_{\loc_0,r_0}$ (if $(\loc_0,r_0)$
belongs to the kernel), and the successors of a node of $T$ are then
recursively defined by its output transitions. When a state~$(\loc,r)$
is reached by an output transition, the child is labelled
by~$\Kernel_{\loc,r}$ if~$(\loc,r)\in \Kernel$, otherwise it is
labelled by~$(\loc,r)$. Edges in~$T$ are labelled by the transitions
used to create them. Along every branch, we stop the construction when
either a final state is reached (i.e.~a state not inside the current
SCC) or the branch contains $3|\rgame|\wmax^e+2\sup|\weightT|+2$ nodes labelled by the same
state ($(\loc,r)$ or~$\Kernel_{\loc,r}$). Since $\rgame$ has a finite
number of states, $T$ is finite. Leaves of $T$ with a location
belonging to~\LocsT are called \emph{target leaves}, others are called
\emph{stopped leaves}.

We now transform~$T$ into a \WTG~\tgame, by replacing every node
labelled by a state~$(\loc,r)$ by a different copy $(\tilde{\loc},r)$
of $(\loc,r)$. Those states are said to inherit from~$(\loc,r)$. Edges
of~$T$ are replaced by the transitions labelling them, and have a
similar notion of inheritance. Every non-leaf node labelled by a
kernel~$\Kernel_{\loc,r}$ is replaced by a copy of the \WTG
$\Kernel_{\loc,r}$, output transitions being plugged in the expected
way.  We deal with stopped leaves labelled by a
kernel~$\Kernel_{\loc,r}$ by replacing them with a single node copy
of~$(\loc,r)$, like we dealt with node labelled by a state~$(\loc,r)$.
State partition between players and weights are inherited from the
copied states of~\rgame.  The only initial state of~\tgame is the
state denoted by $(\tilde{\loc}_0,r_0)$ inherited from $(\loc_0,r_0)$
in the root of $T$ (either $(\loc_0,r_0)$ or $\Kernel_{\loc_0,r_0}$).
The final states of~\tgame are the states derived from leaves of
$T$. If \rgame is a non-negative (resp.~non-positive) SCC, the
output weight function \weightT is inherited from~\rgame on target
leaves and set to $+\infty$ (resp.~$-\infty$) on stopped
leaves.

\subsection{Semi-unfolding correctness}

We will now prove that Proposition~\ref{prop:semi-unfolding} holds
on this \tgame.

\begin{lemma}\label{lm:pathbound}
  All finite plays in~\rgame have cumulated weight (ignoring output
  weights) at least $-|\rgame|\wmax^e$ in the non-negative case, and at most
  $|\rgame|\wmax^e$ in the non-positive case. Moreover, values of the game
  are bounded by $|\rgame|\wmax^e+\sup|\weightT|$.
\end{lemma}
\begin{proof}
  Suppose first that $\rgame$ is a non-negative SCC.  Consider a
  play~\play following a path~\rpath. \rpath can be decomposed into
  $\rpath=\rpath_1\rpath^c_1\cdots\rpath_n\rpath^c_n$ such that
  every~$\rpath^c_i$ is a cycle, and $\rpath_1\dots\rpath_n$ is a
  simple path in \rgame (thus $\sum_{i=1}^n |\rpath_i|\leq|\rgame|$).
  Let us define all plays~$\play_i$ and~$\play^c_i$ as the
  restrictions of~\play on~$\rpath_i$ and~$\rpath^c_i$.  Now, since
  all plays following cycles have cumulated weight at least~$0$,
  $\weightC(\play)=\sum_{i=1}^n \weightC(\play_i)+\weightC(\play^c_i)\geq
  \sum_{i=1}^n -\wmax^e|\play_i| + 0\geq -|\rgame|\wmax^e$. Similarly, we can show
  that every play in a non-positive SCC has cumulated weight at most $|\rgame|\wmax^e$.

  For the bound on the values, consider again two cases. If~\rgame is
  non-negative, consider any memoryless attractor strategy~\stratmin
  for~\MinPl toward~\RStatesF. Since all states have values below $+\infty$,
  all plays obtained from strategies of~\MaxPl will
  follow simple paths of~\rgame, that have cumulated weight at most
  $|\rgame|\wmax^e$ in absolute value. Similarly, if~\rgame is non-positive,
  following the proof of Lemma~\ref{lm:-infty}, since all values are above $-\infty$,
  $\MaxPl$ can ensure $\lnot \phi\Rightarrow
  \mathrm F \RTransF^\R \vee \mathrm F \mathrm G \RTransK$ on all states.
  Then we can construct a strategy~\stratmax for~\MaxPl combining an attractor
  strategy toward~\RStatesF on states satisfying $\mathrm F \RTransF^\R$, a safety
  strategy on states satisfying $G \RTransK$, and an attractor strategy toward the latter
  on all other states.
  Then, all plays obtained from strategies of~\MinPl will
  either not be winning ($G \RTransK$) or follow simple paths of~\rgame.
  Both cases imply that the values of the game
  are bounded by $|\rgame|\wmax^e+\sup|\weightT|$.
\end{proof}

\begin{lemma}\label{lm:rootleafplay}
  All plays in~\tgame from the initial state to a stopped leaf have
  cumulated weight at least $2|\rgame|\wmax^e+2\sup|\weightT|+1$ if the SCC \rgame is
  non-negative, and at most $-2|\rgame|\wmax^e-2\sup|\weightT|-1$ if it is non-positive.
\end{lemma}
\begin{proof}
  Note that by construction, all finite paths in \tgame from the
  initial state to a stopped leaf can be decomposed as
  $\rpath'\rpath_1\cdots\rpath_{3|\rgame|\wmax^e+2\sup|\weightT|+1}$ with all $\rpath_i$
  being cycles. Additionally, those cycles cannot be 0-cycles by
  Proposition~\ref{prop:kernelprop}, since they take at least one
  transition outside of~\Kernel. Therefore the restriction of~\play
  to $\rpath_1\cdots\rpath_{3|\rgame|\wmax^e+1}$ has weight at
  least~$3|\rgame|\wmax^e+2\sup|\weightT|+1$ (in the non-negative case) and at most
  $-3|\rgame|\wmax^e-2\sup|\weightT|-1$ (in the non-positive case). The beginning of the
  play, following $\rpath'$, has cumulated weight at
  least~$-|\rgame|\wmax^e$ (in the non-negative case) and at most
  $|\rgame|\wmax^e$ (in the non-positive case), by Lemma~\ref{lm:pathbound}.
\end{proof}

Two plays
$\play=((\loc_1,r_1),\val_1)\xrightarrow{d_1,\rtrans_1}\cdots
\xrightarrow{d_{n-1},\rtrans_{n-1}}((\loc_n,r_n),\val_n)$ and
$\tilde{\play}=((\tilde{\loc}_1,r_1),\val_1)\xrightarrow{d_1,
  ,\tilde{\rtrans}_1}\cdots \xrightarrow{d_{n-1},
  \tilde{\rtrans}_{n-1}}((\tilde{\loc}_n,r_n),\val_n)$ in $\rgame$ and
$\tgame$, respectively, are said to \emph{mimic} each other if every
$(\tilde{\loc}_i,r_i)$ is inherited from $(\loc_i,r_i)$ and every
transition $\tilde{\rtrans}_i$ is inherited from the transition
$\trans_i$. Combining Lemmas~\ref{lm:rootleafplay} and
\ref{lm:pathbound}, we obtain
\begin{lemma}\label{lm:mimickedplays}
  If \rgame is a non-negative (resp.~non-positive) SCC, every
  play from the initial state and with cumulated weight less than
  $|\rgame|\wmax^e+2\sup|\weightT|+1$ (resp.~greater than $-|\rgame|\wmax^e-2\sup|\weightT|-1$) can be
  mimicked in~\tgame without reaching a stopped leaf. Conversely,
  every play in~\tgame reaching a target leaf can be mimicked
  in~\rgame.
\end{lemma}
\begin{proof}
  We prove only the non-negative case. Let~\play be a play of~\rgame
  with cumulated weight less than $|\rgame|\wmax^e+2\sup|\weightT|+1$.  Consider
  the branch of the unfolded game it follows. If~\play cannot be
  mimicked in~\tgame, then a prefix of~\play reaches the stopped leaf
  of that branch when mimicked in~\tgame.  In this situation, \play
  starts by a prefix of weight at least $2|\rgame|\wmax^e+2\sup|\weightT|+1$ by
  Lemma~\ref{lm:rootleafplay} and then ends with a suffix play of
  weight at least~$-|\rgame|\wmax^e$ by Lemma~\ref{lm:pathbound}, and that
  contradicts the initial assumption. The non-positive case is proved
  exactly the same way, and the converse is true by construction.
\end{proof}

Then, the plays of~\rgame starting in an initial
configuration that cannot be mimicked in~\tgame are not useful for
value computation, which is formalised by Proposition~\ref{prop:unfolding}:
\begin{proposition}\label{prop:unfolding}
 For all valuations $\val_0\in r_0$,
 $\Val_{\game}(\loc_0,\val_0) = \Val_{\tgame}((\tilde{\loc}_0,r_0),\val_0)$.
\end{proposition}
\begin{proof}
  By Lemma~\ref{lm:region-game}, we already know that
  $\Val_\game(\loc_0,\val_0) = \Val_{\rgame}((\loc_0,r_0),\val_0)$.
  Recall that we only left finite values in \rgame (in the final
  weight functions, in particular), and more precisely
  $|\Val_{\rgame}((\loc_0,r_0),\val_0)|\leq |\rgame|\wmax^e+\sup|\weightT|$
  by Lemma~\ref{lm:pathbound}.
  We first show that the value is also finite in \tgame.
  Indeed, if $\Val_{\tgame}((\tilde{\loc}_0,r_0),\val_0)=+\infty$,
  since we assumed all output weights of \rgame bounded,
  we are necessarily in the non-negative case, and $\MaxPl$ is able to ensure
  stopped leaves reachability.

  \paragraph{Claim 1}%
  If $\Val_{\tgame}((\tilde{\loc}_0,r_0),\val_0)=+\infty$, then
  there are no winning strategies in~\rgame for~\MinPl ensuring weight less than
  $|\rgame|\wmax^e+\sup|\weightT|+1$ from~$(\loc_0,r_0)$.

  Thus, we can obtain the contradiction $\Val_{\rgame}((\loc_0,r_0),\val_0)>
  |\rgame|\wmax^e+\sup|\weightT|$.
  \begin{proof}[Proof of Claim 1]%
    By contradiction, consider a strategy \stratmin of \MinPl ensuring weight
    $A\leq|\rgame|\wmax^e+\sup|\weightT|+1$ in \rgame.
    Then, for all \stratmax, the cumulated weight of\\
    $\outcomes_{\rgame}(((\tilde{\loc}_0,r_0),\val_0),\stratmin,\stratmax)$
     (reaching target configuration $(\loc,\val)$)
    is at most $A-\weightT(\loc,\val)\leq |\rgame|\wmax^e+2\sup|\weightT|+1$, and by Lemma~\ref{lm:mimickedplays}
    this play does not reach a stopped leaf when mimicked in \tgame, which is absurd.
  \end{proof}
  If $\Val_{\tgame}((\tilde{\loc}_0,r_0),\val_0)=-\infty$,
  we are necessarily in the non-positive case, and by construction
  this implies having~\MinPl ensuring stopped leaves reachability
  in~\tgame.

  \paragraph{Claim 2}%
  If $\Val_{\tgame}((\tilde{\loc}_0,r_0),\val_0)=-\infty$, then
  there are no winning strategies in~\rgame for~\MaxPl ensuring weight above
  $-|\rgame|\wmax^e-\sup|\weightT|-1$ from~$(\loc_0,r_0)$.

  Thus, we can obtain the contradiction $\Val_{\rgame}((\loc_0,r_0),\val_0)<
  -|\rgame|\wmax^e-\sup|\weightT|$.
  \begin{proof}[Proof of Claim 2]
    By contradiction, consider a strategy \stratmax of \MaxPl ensuring weight
    $A\geq-|\rgame|\wmax^e-\sup|\weightT|-1$ in \rgame.
    Then, for all \stratmin, the cumulated weight of \\ $\outcomes_{\rgame}(((\tilde{\loc}_0,r_0),\val_0),\stratmin,\stratmax)$ (reaching target configuration $(\loc,\val)$)
    is at least $A-\weightT(\loc,\val)\geq -|\rgame|\wmax^e -2\sup|\weightT| -1$, and by Lemma~\ref{lm:mimickedplays}
    this play does not reach a stopped leaf when mimicked in \tgame, which is absurd.
  \end{proof}

  Then, strategies and plays of~\tgame starting from~$(\tilde{\loc}_0,r_0)$
  can be mimicked in~\rgame, therefore
  $\Val_{\rgame}((\loc_0,r_0),\val_0) \leq \Val_\tgame(\tilde
  s_0,\val_0)$:
  If~\rgame is non-negative, for all $\varepsilon>0$ we can fix an
  $\varepsilon$-optimal strategy~\stratmin for~\MinPl in~\tgame.
  It is a winning strategy, so every play derived from \stratmin in \tgame reaches a target leaf,
  and can be mimicked in \rgame by Lemma~\ref{lm:mimickedplays}. Therefore,
  \stratmin can be mimicked in \rgame, where it is also winning, with the same weight.
  From this we deduce $\Val_{\rgame}((\loc_0,r_0),\val_0) \leq \Val_\tgame(\tilde
  s_0,\val_0)$. If~\rgame is non-positive, the same reasoning applies by considering
  an $\varepsilon$-optimal strategy for $\MaxPl$ in~$\tgame$.

  Let us now show that
  $\Val_\tgame((\tilde{\loc}_0,r_0),\val_0) \leq
  \Val_{\rgame}((\loc_0,r_0),\val_0)$.
  If~\rgame is non-negative, let us fix $0<\varepsilon<1$, an
  $\varepsilon$-optimal strategy~\stratmin for~\MinPl in~\rgame, and a
  strategy~\stratmax of~\MaxPl in~\rgame.
  Let $\play$ be their outcome $\outcomes_{\rgame}(((\loc_0,r_0),\val_0),\stratmin,\stratmax))$,
  $\play_k$ be the finite prefix of $\play$ defining its cumulative weight
  and $(\loc_k,\val_k)$ be the configuration defining its output weight, such that
  $\weight_{\rgame}(\play)=\weightC(\play_k)+\weightT(\loc_k,\val_k)$.
  Then, $\weight_{\rgame}(\play)\leq \Val_{\rgame}((\loc_0,r_0),\val_0)+\varepsilon <
  |\rgame|\wmax^e+\sup|\weightT|+1$, therefore
  $\weightC(\play_k)<|\rgame|\wmax^e+\sup|\weightT|+1-\weightT(\loc_k,\val_k)
  \leq |\rgame|\wmax^e+2\sup|\weightT|+1$ and by Lemma~\ref{lm:mimickedplays}
  all such plays $\play$ can be mimicked in~\tgame, and
  $\Val_\tgame((\tilde{\loc}_0,r_0),\val_0) \leq
  \Val_{\rgame}((\loc_0,r_0),\val_0)$.
  Once again, if~\rgame is non-positive, the same reasoning applies by considering
  an $\varepsilon$-optimal strategy for $\MaxPl$ in~$\rgame$.
\end{proof}

This proof not only holds on an SCC, but also on full almost-divergent
\WTG{s}, by simply stacking the semi-unfoldings of each SCC on top of each others.

Note that the semi-unfolding procedure of an SCC depends on $\sup|\weightT|$,
where $\weightT$ can be the value function of an SCCs under the current one.
Assuming all configurations have finite value, we can extend the reasoning of Lemma~\ref{lm:pathbound}
and bound all values in the full game by $|\rgame|\wmax^e+\sup|\weightT|$, which
let us bound uniformly the unfolding depth of each SCC and gives us
a bound on the depth of the complete semi-unfolding tree:
$|\rgame|(5|\rgame|\wmax^e+2\sup|\weightT|+2)+1$

\section{Proofs of the approximation schema  (Section~\ref{sec:approx})}\label{app:approx}

\subsection{Proofs of the approximation of kernels}

\begin{proof}[Proof of Lemma~\ref{lm:close-plays}]
  Since $\play$ and $\play'$ follow the same locations $\loc$ of
  $\game$, one reaches a target location if and only if the other
  does. In the case where they do not reach a target location, both
  weights are infinite, and thus equal. We now look at the case where
  both plays reach a target location, moreover in the same step.

  Consider the region path $\rpath$ of the run $\play$: $\rpath$ can
  be decomposed into a simple path with maximal cycles in it. The
  number of such maximal cycles is bounded by
  $|\Locs\times\regions \Clocks \clockbound|$ and the remaining simple
  path has length at most $|\Locs\times\regions \Clocks
  \clockbound|$. Since all cycles of a kernel are $0$-cycles, the
  parts of $\play$ that follow the maximal cycles have weight exactly
  0.

  Consider the same decomposition for the play $\play'$. Cycles of
  $\rpath$ do not necessarily map to cycles over locations of
  $\Ncgame N$, since the $1/N$-regions could be distinct. However,
  Lemma~\ref{lm:cornerabstract} shows that, for all those cycles of
  $\rpath$, there exists a sequence of finite plays of $\game$ whose
  weight tends to the weight of $\play'$. Since all those finite plays follow a cycle
  of the region game $\rgame$ (with $\game$ being a kernel), they all
  have weight 0. Hence, the parts of $\play'$ that follow the maximal
  cycles of $\rpath$ have also weight exactly 0.

  Therefore, the difference
  $|\weight_\game(\play)-\weight_{\Ncgame N}(\play')|$ is concentrated
  on the remaining simple path of $\rpath$: on each transition of this
  path, the maximal weight difference is $1/N\times \wmax^\Locs$ since $1/N$
  is the largest difference possible in time delays between plays that
  stay $1/N$-close (since they stay in the same
  $1/N$-regions). Moreover, the difference between the output weight
  functions is bounded by $\lipconst/N$, since the output weight function
  $\weightT$ is \lipconst-Lipschitz-continuous and the output weight
  function of $\Ncgame N$ is obtained as limit of $\weightT$.
  Summing the two contributions, we obtain as upper bound
  the constant $\myconst/N$.
\end{proof}

\begin{proof}[Proof of Lemma~\ref{lm:bisimulation}]
  Let us prove that both $\Val_\game(\loc,\val)\leq \Val_{\Ncgame N}((\loc,r,v),v)+\alpha$
  and
  $\Val_{\Ncgame N}((\loc,r,v),v)\leq \Val_\game(\loc,\val)+\alpha$,
  with $\alpha = \myconst/N$. By definition and determinacy
  of turn-based \WTG, this is equivalent to proving these two inequalities:
  \[\inf_{\minstrategy}\sup_{\maxstrategy}
    \weight_\game(\outcomes((\loc,\val),\maxstrategy,\minstrategy))
    \leq \inf_{\minstrategy'}\sup_{\maxstrategy'} \weight_{\Ncgame
      N}(\outcomes(((\loc,r,v),v),\maxstrategy',\minstrategy'))
    +\alpha\]
  \[\sup_{\maxstrategy'}\inf_{\minstrategy'} \weight_{\Ncgame
    N}(\outcomes(((\loc,r,v),v),\maxstrategy',\minstrategy'))
    \leq \sup_{\maxstrategy}\inf_{\minstrategy}
    \weight_\game(\outcomes((\loc,\val),\maxstrategy,\minstrategy))
    +\alpha\]
  Let $(\beta)$ denote $|\weight_\game(\outcomes((\loc,\val),\maxstrategy,\minstrategy))
  - \weight_{\Ncgame
    N}(\outcomes(((\loc,r,v),v),\maxstrategy',\minstrategy'))|\leq
  \alpha$.
  To show the first inequality, it suffices to show that
  for all $\minstrategy'$, there exists $\minstrategy$ such that for
  all $\maxstrategy$, there is $\maxstrategy'$ verifying $(\beta)$.
  For the second, it suffices to show that
  for all $\maxstrategy'$, there exists $\maxstrategy$ such that for
  all $\minstrategy$, there is $\minstrategy'$ verifying $(\beta)$.
  We will detail the proof for the first, the second being syntactically the same,
  with both players swapped.

  Equation $(\beta)$ can be obtained from Lemma~\ref{lm:close-plays},
  under the condition that
  the plays \\
  $\outcomes((\loc,\val),\maxstrategy,\minstrategy)$ and
  $\outcomes(((\loc,r,v),v),\maxstrategy',\minstrategy')$ are
  $1/N$-close.
  Therefore, we fix a strategy $\minstrategy'$ of \MinPl in the game
  $\Ncgame N$, and we construct a strategy $\minstrategy$ of \MinPl in
  $\game$, as well as two mappings
  $f\colon \FPlaysMin_\game\to \FPlaysMin_{\Ncgame N}$ and
  $g\colon \FPlaysMax_{\Ncgame N}\to \FPlaysMax_\game$ such~that:
  \begin{itemize}
  \item for all $\play\in \FPlaysMin_\game$, $\play$ and $f(\play)$
    are $1/N$-close, and if $\play$ is consistent with $\minstrategy$
    and starts in $(\loc,\val)$, then $f(\play)$ is consistent with
    $\minstrategy'$ and starts in $((\loc,r,v),v)$;
  \item for all $\play'\in \FPlaysMax_{\Ncgame N}$, $g(\play')$ and
    $\play'$ are $1/N$-close, and if $\play'$ is consistent with
    $\minstrategy'$ and starts in $((\loc,r,v),v)$, then $g(\play')$
    is consistent with $\minstrategy$ and starts in $(\loc,\val)$.
  \end{itemize}
  We build $\minstrategy$, $f$, and $g$ by induction on the length $n$ of plays, over
  prefixes of plays of length $n-1$, $n$ and $n$, respectively.
  For $n=0$ (plays of length 0 are those restricted to a single
  configuration), we let $f(\loc,\val)=((\loc,r,v),v)$ and $g((\loc,r,v),v)=(\loc,\val)$,
  leaving the other values arbitrary (since we will not use them).

  Then, we suppose $\minstrategy$, $f$,
  and $g$ built until length $n-1$, $n$ and $n$, respectively
  (if $n=0$, $\minstrategy$ has not been build yet), and we define
  them on plays of length $n$, $n+1$ and $n+1$, respectively.
  For every $\play\in \FPlaysMin_\game$ of length $n$, we note
  $\play'=f(\play)$. Consider the decision
  $(d',\trans')=\minstrategy'(\play')$ and $\play'_+$ the prefix
  $\play'$ extended with the decision $(d',\trans')$. By timed
  bisimulation, there exists $(d,\trans)$ such that the prefix
  $\play_+$ composed of $\play$ extended with the decision
  $(d,\trans)$ builds $1/N$-close plays $\play_+$ and $\play'_+$.
  We let $\minstrategy(\play)=(d,\trans)$.
  If $\play_+\in \FPlaysMin_\game$, we also let $f(\play_+)=\play'_+$,
  and otherwise we let $g(\play'_+)=\play_+$.
  Symmetrically,consider $\play'\in\FPlaysMax_{\Ncgame N}$ of length $n$, and
  $\play=g(\play')$. For all possible decisions $(d',\trans')$, by
  timed bisimulation, there exists a decision $(d,\trans)$ in the
  prefix $\play$ such that the respective extended plays $\play'_+$
  and $\play_+$ are $1/N$-close. We then let $g(\play'_+)=\play_+$ if
  $\play_+\in\FPlaysMax_\game$ and $f(\play_+)=\play'_+$ otherwise.
  We extend the definition of $f$ and $g$ arbitrarily for other
  prefixes of plays. The properties above are then trivially verified.

  We then fix a strategy $\maxstrategy$ of \MaxPl in the game $\game$,
  which determines a unique play\\
  $\outcomes((\loc,\val),\maxstrategy,\minstrategy)$. We construct a
  strategy $\maxstrategy'$ of \MaxPl in the game $\Ncgame N$ by
  building the unique play
  $\outcomes(((\loc,r,v),v),\maxstrategy',\minstrategy')$ we will be
  interested in, such that each of its prefixes is in relation, via
  $f$ or $g$, to the associated prefix of
  $\outcomes((\loc,\val),\maxstrategy,\minstrategy)$. Thus, we only
  need to consider a prefix of play $\play'\in \FPlaysMax_{\Ncgame N}$
  that starts in $((\loc,r,v),v)$ and is consistent with
  $\minstrategy'$, and $\maxstrategy'$ built so far. Consider the play
  $\play = g(\play')$, starting in $(\loc,\val)$ and consistent with
  $\minstrategy$, and $\maxstrategy$ (by assumption). For the decision
  $(d,\trans)=\maxstrategy(\play)$ (letting $\play_+$ be the extended
  prefix), the definition of $f$ and $g$ ensures that there exists a
  decision $(d',\trans')$ after $\play'$ that results in an extended
  play $\play'_+$ that is $1/N$-close, via $f$ or $g$, with
  $\play_+$. We thus can choose $\maxstrategy'(\play')=(d',\trans')$.

  We finally have built two plays
  $\outcomes((\loc,\val),\maxstrategy,\minstrategy)$ and
  $\outcomes((\loc',\val'),\maxstrategy',\minstrategy')$ that are
  $1/N$-close, as needed, which concludes this proof.
\end{proof}

\begin{proof}[Proof of Lemma~\ref{lm:kernel-approx-regular}]
  By construction, the approximated value is piecewise linear with one
  piece per $1/N$-region. To prove the Lipschitz constant, it is then
  sufficient to bound the difference between
  $\Val_{\Ncgame N}((\loc,r,v),v)$ and
  $\Val_{\Ncgame N}((\loc,r,v'),v')$, for $v$ and $v'$ two corners of
  a $1/N$-region $r$. We can  pick any valuation \val in $r$ and
  apply Lemma~\ref{lm:bisimulation} twice, between \val and $v$, and
  between \val and $v'$. %
  We obtain $|\Val_{\Ncgame N}((\loc,r,v),v)-\Val_{\Ncgame N}((\loc,r,v'),v')|
  \leq 2  \myconst /N = 2 \| v-v'\|_\infty \myconst$.
\end{proof}

\subsection{Computing the value of an acyclic WTG}\label{app:value-ite}

Note that for a piecewise linear functions with finitely many pieces,
being \lipconst-Lipschitz-continuous over regions is equivalent to being
continuous over regions and having all partial derivatives bounded
by \lipconst in absolute value.

\begin{lemma}\label{lm:value-ite-partial-derivatives}
  If for all $\loc\in\Locs$, $V_\loc$ is piecewise linear with
  finitely many pieces %
  that have all their partial derivatives bounded
  by \lipconst in absolute value,
  then for all $\loc\in\Locs$,
  $\mathcal{F}(V)_\loc$ is continuous over regions and piecewise linear with
  finitely many pieces %
  that have all their partial derivatives bounded
  by $\max(\lipconst,|\weight(\loc)| + (n-1) \lipconst)$ in absolute value.
\end{lemma}
\begin{proof}
  We will show that for every region $r$, $\mathcal{F}(V)$ restricted
  to $r$ has those properties. Note that they are transmitted over finite $\min$
  and $\max$ operations.
  The continuity over regions is easy to prove because it is stable by $\inf$ and $\sup$.
  We now use the notations and definitions of
  \cite{ABM04} to bound the partial derivatives.%
  There exists a partition cost function $(P,F)$ that
  represents $V$, with $P$ an $n$-dimensional nested tube partition
  and $F$ a mapping from the leaf nodes of $P$ to linear expressions
  over variables in \Clocks. Intuitively, $P$ defines a finite
  arborescence of convex spaces, defined by linear inequalities, whose
  root is the whole region $r$ and whose leaves partition $r$ into
  \emph{cells}. A crucial property of those cells
  (\cite[Theorem~4]{ABM04}) is that, for a given valuation $\val$, the
  delays $t$ that need to be considered in the $\sup$ or $\inf$
  operation of $\mathcal{F}(V)_{(\loc,\val)}$ correspond to the
  intersection points of the diagonal half line containing the time
  successors of $\val$ and borders of cells (if $\val^b$ is such a
  valuation, $t=\|\val^b-\val\|_\infty$ is the associated delay).
  In particular, there is a finite number of such borders, and the final
  $\mathcal{F}(V)_\loc$ function can be written as a finite nesting of
  finite $\min$ and $\max$ operations over linear terms, %
  each corresponding to a choice of delay and a transition to take.
  Formally, there are several cases to consider to define those terms,
  depending on delay and transition choices. For each available transition $\trans$,
  those terms can either be:
  \begin{enumerate}
    \item If a delay $0$ is taken and all clocks in $Y\subseteq\Clocks$ are reset by $\trans$,
    then \\ $\weightC((\loc,\val)\xrightarrow{0}(\loc,\val)\xrightarrow{\trans}(\loc',\val[Y:=0]))=
    \weightC(\trans) + V_{(\loc',\val[Y:=0])}$
    \item If a delay $t>0$ (leading to valuation $\val^b$ on border $B$) is taken
    and the clocks in $Y$ are reset by $\trans$,
    then $\weightC((\loc,\val)\xrightarrow{t}(\loc,\val^b)\xrightarrow{\trans}(\loc',\val^b[Y:=0]))=
    \weightC(\loc)\times t + \weightC(\trans) + V_{(\loc',\val^b[Y:=0])}$
  \end{enumerate}

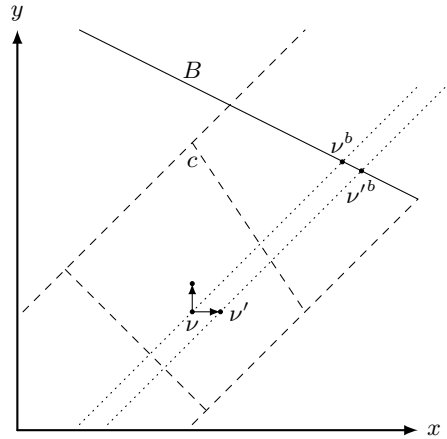
\begin{figure}[ht]
\centering
  \begin{tikzpicture}[scale=0.75,cap=round,>=latex]
      \draw[thick,->] (-1.1cm,-1.1cm) -- (6cm,-1.1cm) node[right,fill=white] {$x$};
      \draw[thick,->] (-1.1cm,-1.1cm) -- (-1.1cm,6cm) node[above,fill=white] {$y$};

      \draw[dotted,-] (0cm,-1cm) -- (6cm,5cm);

      \filldraw[black] (2cm,1cm) circle(1pt);
      \draw (2cm,1cm) node[below] {$\val$};

      \filldraw[black] (2cm,1.5cm) circle(1pt);
      \draw[->] (2cm,1cm) -- (2cm,1.5cm);

      \filldraw[black] (2.5cm,1cm) circle(1pt);
      \draw (2.5cm,1cm) node[right] {$\val'$};%
      \draw[->] (2cm,1cm) -- (2.5cm,1cm);

      \draw[dotted,-] (0.5cm,-1cm) -- (6.5cm,5cm);

      \draw[-] (0cm,6cm) -- (6cm,3cm);
      \draw (2cm,5cm) node[above] {$B$};

      \filldraw[black] (4.66cm,3.66cm) circle(1pt);
      \draw (4.66cm,3.66cm) node[above] {$\val^b$};

      \filldraw[black] (5cm,3.5cm) circle(1pt);
      \draw (5cm,3.5cm) node[below] {${\val'}^b$};

      \draw[dashed,-] (-1cm,1cm) -- (4cm,6cm);
      \draw[dashed,-] (2cm,-1cm) -- (6cm,3cm);
      \draw[dashed,-] (2cm,4cm) -- (4cm,1cm);
      \draw[dashed,-] (-0.25cm,1.75cm) -- (2.25cm,-0.75cm);
      \draw (2cm,3.9cm) node[below] {$c$};
    \end{tikzpicture}
   \caption{A tubular cell $c$ as described in the proof of Lemma~\ref{lm:value-ite-partial-derivatives}.
   Dashed lines bound the cell $c$, dotted lines are proof constructions.
   }\label{fig:value-ite-lipschitz}
\end{figure}

  In the first case, the resulting partial derivatives are $0$ for clocks in $Y$,
  and the same as the partial derivatives in $V_{\loc'}$ for all other clocks,
  which allows us to conclude that they are bounded by \lipconst.
  We now consider the second case.
  We argue that the second case could be decomposed as a delay followed by a transition
  of the first case, meaning that we can assume $Y=\emptyset$ without loss of generality.

  There are again two cases: the border $B$ being inside a region or on the frontier of a region.

  If the border is not the frontier of a region, it is the
  intersection points of two affine pieces of $V_{\loc'}$ whose
  equations (in the space $\R^{n+1}$ whose $n$ first coordinates are
  the clocks $(x_1,\ldots,x_n)$ and the last coordinate correspond to the value
  $V_{\loc'}(x_1,\ldots,x_n)$) can be written
  $y=\sum_{i=1}^n a_ix_i+b$ (before the border) and
  $y=\sum_{i=1}^n a'_ix_i+b'$ (after the border). Therefore, valuations
  of the borders all fulfil the equation
  \begin{equation}
    \sum_{i=1}^n(a'_i-a_i)x_i + b-b'=0\label{eq:1}
  \end{equation}
  We let $A=\sum_{i=1}^n(a'_i-a_i)$. Consider that $\loc$ is a
  location of \MinPl (the very same reasoning applies to the case of a
  location of \MaxPl). Since $\mathcal{F}$ computes an infimum, we
  know that the function mapping the delay $t$ to the weight obtained
  from reaching $\val+t$ is decreasing before the
  border and increasing after. These functions are locally affine
  which implies that their slopes verify:
  \begin{equation}
    \weight(\loc)+\sum_{i=1}^n a_i\leq 0\quad \text{ and } \quad
    \weight(\loc)+\sum_{i=1}^n a'_i\geq 0\,.\label{eq:2}
  \end{equation}
  We deduce from these two inequalities that $A\geq 0$. The case where
  $A=0$ would correspond to the case where the border contains a
  diagonal line, which is forbidden, and $A>0$.
  Consider now a valuation of coordinates $\val=(x_1,\ldots,x_n)$ and
  another valuation of coordinates
  $\val'=(x_1,\ldots,x_{k-1},x_k+\lambda,x_{k+1},\ldots,x_n)$. The
  delays $t$ and $t'$ needed to arrive to the border starting from
  these two valuations are such that $\val+t$ and $\val'+t'$ both
  verify \eqref{eq:1}. We can then deduce that
  $t'-t=\lambda\frac{a_k-a'_k}{A}$. It is now possible to compute the
  partial derivative of $\mathcal{F}(V)_\loc$ in the $k$-th coordinate
  using
  \[\frac{\mathcal{F}(V)_{\loc,\val'}-
      \mathcal{F}(V)_{\loc,\val}}{\lambda} = \frac{\weight(\loc)(t'-t')+
      V_{\loc',\val'+t'}-V_{\loc',\val+t}}{\lambda}.\] We may compute
  it by using the equations of the affine pieces before or after the
  border. We thus obtain
  \begin{align*}
    \frac{\mathcal{F}(V)_{\loc,\val'}-\mathcal{F}(V)_{\loc,\val}}{\lambda}
    &=
      \frac{a_k-a'_k} A (\weight(\loc)+\sum_{i=1}^n a_i) + a_k\\
    \frac{\mathcal{F}(V)_{\loc,\val'}-\mathcal{F}(V)_{\loc,\val}}{\lambda}
    &=
      \frac{a_k-a'_k} A (\weight(\loc)+\sum_{i=1}^n a'_i) + a'_k
  \end{align*}
  In the case where $a_k\geq a'_k$, the first equation, with
  \eqref{eq:2}, allows us to obtain that the partial derivative is at
  most $a_k$. We may then lower $\weight(\loc)$ by $-\sum_{i=1}^n a'_i$
  to obtain that the partial derivative is at least $a'_k$. Since
  $a_k$ and $a'_k$ are bounded in absolute value by \lipconst, so is the
  partial derivative. We get the same result by reasoning on the
  second equation if $a'_k\geq a_k$.

  We now come back to the case where the border is on the frontier of a region.
  Then, it is a segment of a line of equation $x_k=c$ for some $k$ and $c$.
  $V_{\loc'}$ contains at most three values for points of $B$: The limit coming from before the border,
  the value at the border, and the limit coming from after the border.
  The computation of $\mathcal{F}(V)$ computes values obtained from all three
  and takes the $\min$ (or the $\max$).

  Now, let $y=\sum_{i=1}^n a_ix_i+b$ be the equation
  defining the linear piece of $V_{\loc'}$ before the border (resp.~at the border, after the border).
  Consider now a valuation of coordinates $\val=(x_1,\ldots,x_n)$ and
  another valuation of coordinates
  $\val'=(x_1,\ldots,x_{j-1},x_j+\lambda,x_{j+1},\ldots,x_n)$. The
  delays $t$ and $t'$ needed to arrive to the border starting from
  these two valuations are such that $\val+t$ and $\val'+t'$ both
  verify $x_k=c$. We can then deduce that
  $t'-t=0$ if $j\neq k$ and $t'-t=-\lambda$ if $j=k$. It is now possible to compute the
  partial derivative of $\mathcal{F}(V)_\loc$ in the $j$-th coordinate
  using
  \[\frac{\mathcal{F}(V)_{\loc,\val'}-
      \mathcal{F}(V)_{\loc,\val}}{\lambda} = \frac{\weight(\loc)(t'-t')+
      V_{\loc',\val'+t'}-V_{\loc',\val+t}}{\lambda}.\] We may compute
  it by using the equations of the linear piece before the
  border (resp.~at the border, after the border).
  Then, $V_{\loc',\val+t}=\sum_{i=1}^n a_i(x_i+t) +b=(\sum_{i=1\neq k}^n a_i(x_i+t)) + a_k c+b +$
  and $V_{\loc',\val'+t'}=(\sum_{i=1\neq k}^n a_i(x_i+t')) + a_k c +b$.
  We thus obtain
  \begin{align*}
    \frac{\mathcal{F}(V)_{\loc,\val'}-\mathcal{F}(V)_{\loc,\val}}{\lambda}
    &=
      a_j\text{ if }j\neq k\\
    \frac{\mathcal{F}(V)_{\loc,\val'}-\mathcal{F}(V)_{\loc,\val}}{\lambda}
    &=
      -\weight(\loc)-\sum_{i=1,i\neq k}^n a_i\text{ otherwise}
  \end{align*}

  Then, the partial derivatives are bounded, in absolute value, by $|\weight(\loc)|+(n-1) \lipconst$.

\end{proof}

As a corollary, we can now obtain Lemma~\ref{lm:tree-lipschitz}, or more precisely:
\begin{lemma}\label{lm:tree-lipschitz-detailed}
  Consider an acyclic WTG \game of depth $i$ with all the output
  weights being
  \lipconst-Lipschitz-continuous over each region
  (and piecewise linear, with finitely many pieces).
  Then,
  \begin{itemize}
  \item if $|\Clocks|=1$, $\Val_\game=\Val^i_\game$ is
    $\max(\lipconst,\wmax^\Locs)$-Lipschitz-continuous over regions;
  \item if $|\Clocks|=2$, $\Val_\game=\Val^i_\game$ is
    $(i*\wmax^\Locs+\lipconst)$-Lipschitz-continuous over regions;
  \item otherwise, $\Val_\game=\Val^i_\game$ is
    $(\wmax^\Locs\frac{(|\Clocks|-1)^i-1}{|\Clocks|-2}+(|\Clocks|-1)^i
    \lipconst)$-Lipschitz-continuous over regions.
  \end{itemize}

\end{lemma}

\subsection{Example of an execution of the approximation schema}\label{app:example}

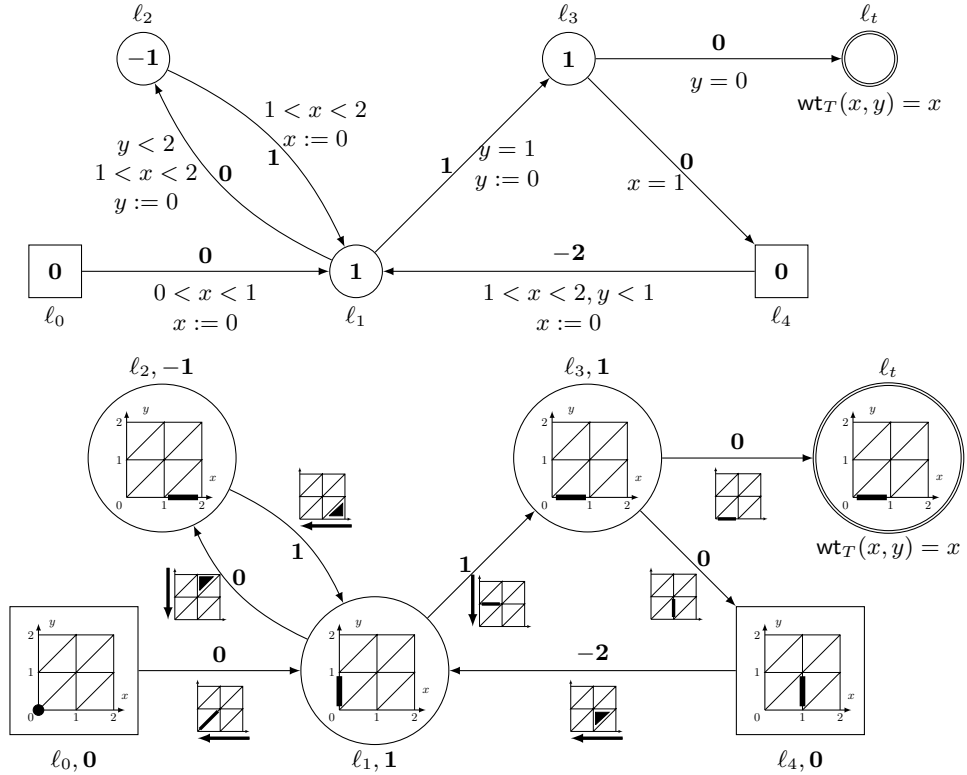
\begin{figure}[ht]
  \centering
\begin{tikzpicture}[node distance=4cm,auto,->,>=latex]
  \node[player2](0){\makebox[0mm][c]{$\mathbf{0}$}};
  \node()[below of=0,node distance=6mm]{$\loc_0$};
  \node[player1](1)[right of=0]{\makebox[0mm][c]{$\mathbf{1}$}};
  \node()[below of=1,node distance=6mm]{$\loc_1$};
  \node[player1](2)[above left of=1]{\makebox[0mm][c]{$\mathbf{-1}$}};
  \node()[above of=2,node distance=6mm]{$\loc_2$};
  \node[player1](3)[above right of=1]{\makebox[0mm][c]{$\mathbf{1}$}};
  \node()[above of=3,node distance=6mm]{$\loc_3$};
  \node[player2](4)[below right of=3]{\makebox[0mm][c]{$\mathbf{0}$}};
  \node()[below of=4,node distance=6mm]{$\loc_4$};
  \node[player1](5)[right of=3,accepting]{\makebox[0mm][c]{$\mathbf{}$}};
  \node()[above of=5,node distance=6mm]{$\loc_t$};
  \node()[below of=5,node distance=6mm]{$\weightT(x,y)=x$};

  \path
  (0) edge node[below]{$\begin{array}{c} 0<x<1\\ x:=0 \end{array}$} node[above]{$\mathbf{0}$} (1)
  (1) edge[bend left=20] node[left,xshift=-2mm,yshift=2mm]{$\begin{array}{c} y<2\\1<x<2\\ y:=0 \end{array}$} node[above]{$\mathbf{0}$} (2)
  (2) edge[bend left=20] node[right,xshift=-3mm,yshift=2mm]{$\begin{array}{c} 1<x<2\\ x:=0 \end{array}$} node[below]{$\mathbf{1}$} (1)
  (1) edge node[right]{$\begin{array}{c} y=1\\ y:=0 \end{array}$} node[left]{$\mathbf{1}$} (3)
  (3) edge node[below,xshift=-2mm]{$\begin{array}{c} x=1 \end{array}$} node[right]{$\mathbf{0}$} (4)
  (4) edge node[below]{$\begin{array}{c} 1<x<2, y<1\\x:=0 \end{array}$} node[above]{$\mathbf{-2}$} (1)
  (3) edge node[below]{$\begin{array}{c} y=0 \end{array}$} node[above]{$\mathbf{0}$} (5);
\end{tikzpicture}
  \\
  \centering
\begin{tikzpicture}[node distance=4cm,auto,->,>=latex]

  \node[draw,regular polygon,regular polygon sides=4,inner sep=-2mm,minimum size=10mm](0){\makebox[0mm][c]{
  \scalebox{0.5}{\begin{tikzpicture}
    \path[draw,->,thick](0,0) -> (2.3,0) node[above] {$x$};
    \path[draw,->,thick](0,0) -> (0,2.3) node[right] {$y$};
    \path[draw,-] (0,0) -- (2,2)
    (1,0) -- (1,2) -- (0,1) -- (2,1) -- (1,0)
    (0,2) -- (2,2) -- (2,0);
    \node () at (1,-0.2) {$1$};
    \node () at (2,-0.2) {$2$};
    \node () at (-0.2,-0.2) {$0$};
    \node () at (-0.2,1) {$1$};
    \node () at (-0.2,2) {$2$};
    \node[draw, fill=black,circle,inner sep=1mm, minimum size=1mm] () at (0,0) {};
  \end{tikzpicture}}}};
  \node()[below of=0,node distance=12mm]{$\loc_0, \mathbf{0}$};

  \node[player1,minimum size=10mm](1)[right of=0]{\makebox[0mm][c]{
  \scalebox{0.5}{\begin{tikzpicture}
    \path[draw,->,thick](0,0) -> (2.3,0) node[above] {$x$};
    \path[draw,->,thick](0,0) -> (0,2.3) node[right] {$y$};
    \path[draw,-] (0,0) -- (2,2)
    (1,0) -- (1,2) -- (0,1) -- (2,1) -- (1,0)
    (0,2) -- (2,2) -- (2,0);
    \node () at (1,-0.2) {$1$};
    \node () at (2,-0.2) {$2$};
    \node () at (-0.2,-0.2) {$0$};
    \node () at (-0.2,1) {$1$};
    \node () at (-0.2,2) {$2$};
    \path[draw, -, fill=black, line width=1.5mm] (0,0.1) -- (0,0.9);
  \end{tikzpicture}}}};
  \node()[below of=1,node distance=12mm]{$\loc_1, \mathbf{1}$};

  \node[player1,minimum size=10mm](2)[above left of=1]{\makebox[0mm][c]{
  \scalebox{0.5}{\begin{tikzpicture}
    \path[draw,->,thick](0,0) -> (2.3,0) node[above] {$x$};
    \path[draw,->,thick](0,0) -> (0,2.3) node[right] {$y$};
    \path[draw,-] (0,0) -- (2,2)
    (1,0) -- (1,2) -- (0,1) -- (2,1) -- (1,0)
    (0,2) -- (2,2) -- (2,0);
    \node () at (1,-0.2) {$1$};
    \node () at (2,-0.2) {$2$};
    \node () at (-0.2,-0.2) {$0$};
    \node () at (-0.2,1) {$1$};
    \node () at (-0.2,2) {$2$};
    \path[draw, -, fill=black, line width=1.5mm] (1.1,0) -- (1.9,0);
  \end{tikzpicture}}}};
  \node()[above of=2,node distance=12mm]{$\loc_2, \mathbf{-1}$};

  \node[player1,minimum size=10mm](3)[above right of=1]{\makebox[0mm][c]{
  \scalebox{0.5}{\begin{tikzpicture}
    \path[draw,->,thick](0,0) -> (2.3,0) node[above] {$x$};
    \path[draw,->,thick](0,0) -> (0,2.3) node[right] {$y$};
    \path[draw,-] (0,0) -- (2,2)
    (1,0) -- (1,2) -- (0,1) -- (2,1) -- (1,0)
    (0,2) -- (2,2) -- (2,0);
    \node () at (1,-0.2) {$1$};
    \node () at (2,-0.2) {$2$};
    \node () at (-0.2,-0.2) {$0$};
    \node () at (-0.2,1) {$1$};
    \node () at (-0.2,2) {$2$};
    \path[draw, -, fill=black, line width=1.5mm] (0.1,0) -- (0.9,0);
  \end{tikzpicture}}}};
  \node()[above of=3,node distance=12mm]{$\loc_3, \mathbf{1}$};

  \node[draw,regular polygon,regular polygon sides=4,inner sep=-2mm,minimum size=10mm](4)[below right of=3]{\makebox[0mm][c]{
  \scalebox{0.5}{\begin{tikzpicture}
    \path[draw,->,thick](0,0) -> (2.3,0) node[above] {$x$};
    \path[draw,->,thick](0,0) -> (0,2.3) node[right] {$y$};
    \path[draw,-] (0,0) -- (2,2)
    (1,0) -- (1,2) -- (0,1) -- (2,1) -- (1,0)
    (0,2) -- (2,2) -- (2,0);
    \node () at (1,-0.2) {$1$};
    \node () at (2,-0.2) {$2$};
    \node () at (-0.2,-0.2) {$0$};
    \node () at (-0.2,1) {$1$};
    \node () at (-0.2,2) {$2$};
    \path[draw, -, fill=black, line width=1.5mm] (1,0.1) -- (1,0.9);
  \end{tikzpicture}}}};
  \node()[below of=4,node distance=12mm]{$\loc_4, \mathbf{0}$};

  \node[player1,minimum size=10mm](5)[right of=3,accepting]{\makebox[0mm][c]{
  \scalebox{0.5}{\begin{tikzpicture}
    \path[draw,->,thick](0,0) -> (2.3,0) node[above] {$x$};
    \path[draw,->,thick](0,0) -> (0,2.3) node[right] {$y$};
    \path[draw,-] (0,0) -- (2,2)
    (1,0) -- (1,2) -- (0,1) -- (2,1) -- (1,0)
    (0,2) -- (2,2) -- (2,0);
    \node () at (1,-0.2) {$1$};
    \node () at (2,-0.2) {$2$};
    \node () at (-0.2,-0.2) {$0$};
    \node () at (-0.2,1) {$1$};
    \node () at (-0.2,2) {$2$};
    \path[draw, -, fill=black, line width=1.5mm] (0.1,0) -- (0.9,0);
  \end{tikzpicture}}}};
  \node()[above of=5,node distance=12mm]{$\loc_t$};
  \node()[below of=5,node distance=12mm]{$\weightT(x,y)=x$};

  \path
  (0) edge node[below]{\makebox[0mm][c]{
  \scalebox{0.3}{\begin{tikzpicture}
    \path[draw,->,thick](0,0) -> (2.3,0);
    \path[draw,->,thick](0,0) -> (0,2.3);
    \path[draw,-] (0,0) -- (2,2)
    (1,0) -- (1,2) -- (0,1) -- (2,1) -- (1,0)
    (0,2) -- (2,2) -- (2,0);
    \path[draw, -, fill=black, line width=1.5mm] (0.1,0.1) -- (0.9,0.9);
    \path[draw,->, fill=black, line width=1.5mm](2.3,-0.3) -> (0,-0.3);
  \end{tikzpicture}}}} node[above]{$\mathbf{0}$} (1)

  (1) edge[bend left=20] node[left,xshift=-5mm]{\makebox[0mm][c]{
  \scalebox{0.3}{\begin{tikzpicture}
    \path[draw,->,thick](0,0) -> (2.3,0);
    \path[draw,->,thick](0,0) -> (0,2.3);
    \path[draw,-] (0,0) -- (2,2)
    (1,0) -- (1,2) -- (0,1) -- (2,1) -- (1,0)
    (0,2) -- (2,2) -- (2,0);
    \path[draw, -, fill=black] (1.1,1.3) -- (1.7,1.9) -- (1.1,1.9) -- (1.1,1.3);
    \path[draw,->, fill=black, line width=1.5mm](-0.3,2.3) -> (-0.3,0);
  \end{tikzpicture}}}} node[above]{$\mathbf{0}$} (2)

  (2) edge[bend left=20] node[right,xshift=2mm,yshift=5mm]{\makebox[0mm][c]{
  \scalebox{0.3}{\begin{tikzpicture}
    \path[draw,->,thick](0,0) -> (2.3,0);
    \path[draw,->,thick](0,0) -> (0,2.3);
    \path[draw,-] (0,0) -- (2,2)
    (1,0) -- (1,2) -- (0,1) -- (2,1) -- (1,0)
    (0,2) -- (2,2) -- (2,0);
    \path[draw, -, fill=black] (1.3,0.15) -- (1.9,0.7) -- (1.9,0.15) -- (1.3,0.15);
    \path[draw,->, fill=black, line width=1.5mm](2.3,-0.3) -> (0,-0.3);
  \end{tikzpicture}}}} node[below]{$\mathbf{1}$} (1)

  (1) edge node[below right,xshift=1mm]{\makebox[0mm][c]{
  \scalebox{0.3}{\begin{tikzpicture}
    \path[draw,->,thick](0,0) -> (2.3,0);
    \path[draw,->,thick](0,0) -> (0,2.3);
    \path[draw,-] (0,0) -- (2,2)
    (1,0) -- (1,2) -- (0,1) -- (2,1) -- (1,0)
    (0,2) -- (2,2) -- (2,0);
    \path[draw, -, fill=black, line width=1.5mm] (0.1,1) -- (0.9,1);
    \path[draw,->, fill=black, line width=1.5mm](-0.3,2.3) -> (-0.3,0);
  \end{tikzpicture}}}} node[left]{$\mathbf{1}$} (3)

  (3) edge node[below,xshift=-2mm]{\makebox[0mm][c]{
  \scalebox{0.3}{\begin{tikzpicture}
    \path[draw,->,thick](0,0) -> (2.3,0);
    \path[draw,->,thick](0,0) -> (0,2.3);
    \path[draw,-] (0,0) -- (2,2)
    (1,0) -- (1,2) -- (0,1) -- (2,1) -- (1,0)
    (0,2) -- (2,2) -- (2,0);
    \path[draw, -, fill=black, line width=1.5mm] (1,0.1) -- (1,0.9);
  \end{tikzpicture}}}} node[right]{$\mathbf{0}$} (4)

  (4) edge node[below]{\makebox[0mm][c]{
  \scalebox{0.3}{\begin{tikzpicture}
    \path[draw,->,thick](0,0) -> (2.3,0);
    \path[draw,->,thick](0,0) -> (0,2.3);
    \path[draw,-] (0,0) -- (2,2)
    (1,0) -- (1,2) -- (0,1) -- (2,1) -- (1,0)
    (0,2) -- (2,2) -- (2,0);
    \path[draw, -, fill=black] (1.1,0.3) -- (1.7,0.85) -- (1.1,0.85) -- (1.1,0.3);
    \path[draw,->, fill=black, line width=1.5mm](2.3,-0.3) -> (0,-0.3);
  \end{tikzpicture}}}} node[above]{$\mathbf{-2}$} (1)

  (3) edge node[below]{\makebox[0mm][c]{
  \scalebox{0.3}{\begin{tikzpicture}
    \path[draw,->,thick](0,0) -> (2.3,0);
    \path[draw,->,thick](0,0) -> (0,2.3);
    \path[draw,-] (0,0) -- (2,2)
    (1,0) -- (1,2) -- (0,1) -- (2,1) -- (1,0)
    (0,2) -- (2,2) -- (2,0);
    \path[draw, -, fill=black, line width=1.5mm] (0.1,0) -- (0.9,0);
  \end{tikzpicture}}}} node[above]{$\mathbf{0}$} (5);

\end{tikzpicture}
  \caption{
  A weighted timed game \game with two clocks $x$ and $y$,
  and the portion of its region game \rgame accessible from configuration $(\loc_0,(0,0))$.
  Locations of \MinPl (resp. \MaxPl) are depicted as circles (resp. squares).
  The states of \rgame are labeled by their associated region, location and weight,
  and transitions are labeled by a representation of their guards and resets.
  Since each location \loc of \game leads to a unique states $(\loc,r)$ of \rgame,
  we will refer to states by their associated location label.}
  \label{fig:example-regions}
\end{figure}

We are given the \WTG \game in Figure~\ref{fig:example-regions} and
$\varepsilon\in\Q_{>0}$, and want to compute
an $\varepsilon$-approximation of its value in location $\loc_0$
for the valuation $(x{=}0,y{=}0)$, denoted $\Val_\game(\loc_0,(0,0))$.
In this example, we will use $\varepsilon{=}15$ because the computations would
not be readable with a smaller precision.
\rgame contains one SCC $\{\loc_1,\loc_2,\loc_3,\loc_4\}$, made of two simple cycles.
$\rpath_1=\loc_1\rightarrow\loc_2\rightarrow\loc_1$ is a positive cycle
(all plays following $\rpath_1$ have cumulated weight in the interval $(1,3)$)
and $\rpath_2=\loc_1\rightarrow\loc_3\rightarrow\loc_4\rightarrow\loc_1$
is a 0-cycle (all plays following $\rpath_2$ have cumulated weight $0$).
This can be checked by Lemma~\ref{lm:cornerabstract}.

Therefore, \rgame only contains non-negative SCCs and is almost-divergent.
Since all states are in the attractor of \MinPl towards $\LocsT$,
all cycles are non-negative and the output weight function is bounded (on all reachable regions),
there are no configurations in \rgame with value $+\infty$ or $-\infty$.

We let the kernel $\Kernel$ be the sub-game of \rgame defined by $\rpath_2$,
and we construct a semi-unfolding \tgame of \rgame of equivalent value.
Following Appendix~\ref{app:unfolding}, we should unfold the game until every
stopped branch contains a state seen at least $3|\rgame|\wmax^e+2\sup|\weightT|+2=3*3*4+2*1=38$ times.
We will unfold with bound $4$ instead of $38$ for readability (it is enough on this example).
Thus the infinite branch $(\loc_1\loc_2)^\omega$ is stopped when $\loc_1$ is reached for the fourth time,
as depicted in Figure~\ref{fig:example-unfolding}.

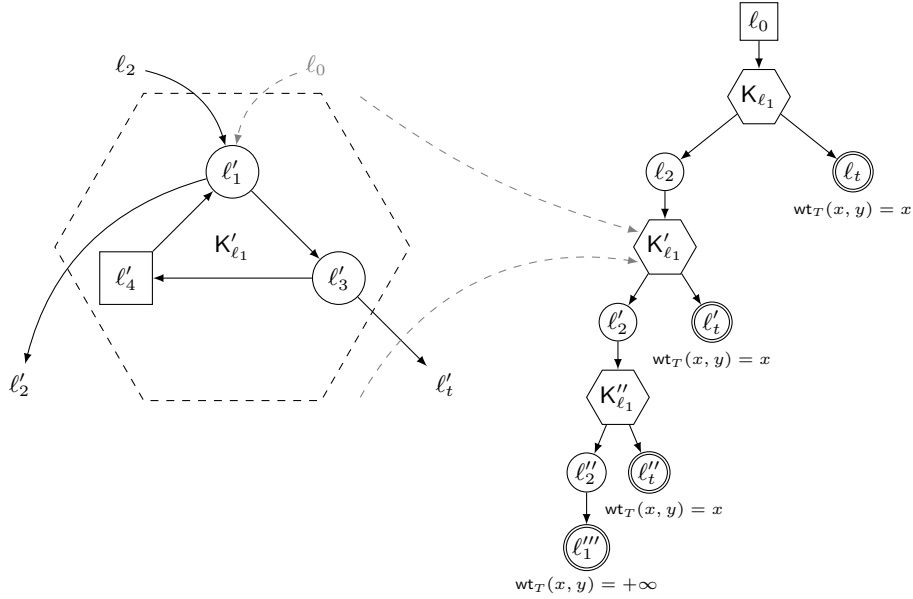
\begin{figure}[ht]
  \centering
  \begin{tikzpicture}[node distance=2cm,auto,->,>=latex]
  \begin{scope}
    \node[player1] (1) {$\loc_1'$};
    \node[player1] (3) [below right of=1] {$\loc_3'$};
    \node[player2] (4) [below left of=1] {$\loc_4'$};
    \node[] (5) [below right of=3] {$\loc_t'$};
    \node[] (2out) [below left of=4] {$\loc_2'$};
    \node[gray] (0) [above right of=1,xshift=-3mm] {$\loc_0$};
    \node[] (2in) [above left of=1] {$\loc_2$};
    \node[draw,dashed,regular polygon,regular polygon sides=6,inner sep=12mm]  [below of=1,yshift=10mm] {$\Kernel_{\loc_1}'$};

    \path
    (1) edge (3)
    (3) edge (4)
    (4) edge (1)
    (0) edge[dashed,gray,bend right=30] (1)
    (2in) edge[bend left=30] (1)
    (1) edge[bend right=30] (2out)
    (3) edge (5)
    ;
  \end{scope}

  \path[draw,gray,dashed] (1.7,1) to[bend right=10] (5.4,-0.8);
  \path[draw,gray,dashed] (1.7,-3) to[bend left=35] (5.4,-1.2);

  \begin{scope}[xshift=7cm,yshift=20mm,every
    node/.style={draw,shape=circle,minimum size=5mm,inner sep=0.3mm},
    level/.style={sibling distance=5cm/#1},level distance=10mm,->]
    \node[rectangle] (0) {$\loc_0$}
    child { node[regular polygon,regular polygon sides=6,inner sep=0.1mm] {$\Kernel_{\loc_1}$}
      child { node[] {$\loc_2$}
        child { node[regular polygon,regular polygon sides=6,inner sep=0.1mm] (kernel) {$\Kernel_{\loc_1}'$}
          child { node[] {$\loc_2'$}
            child { node[regular polygon,regular polygon sides=6,inner sep=0.1mm] {$\Kernel_{\loc_1}''$}
              child { node[] {$\loc_2''$}
                child { node[accepting] (stop) {$\loc_1'''$} }
              }
              child { node[accepting] (out3) {$\loc_t''$} }
            }
          }
          child { node[accepting] (out2) {$\loc_t'$} }
        }
      }
      child { node[accepting] (out1) {$\loc_t$} }
    };
    \node[node distance=5mm,below of=stop,draw=none] () {\smaller[2]{$\weightT(x,y)=+\infty$}};
    \node[node distance=5mm,below of=out1,draw=none] () {\smaller[2]{$\weightT(x,y)=x$}};
    \node[node distance=5mm,below of=out2,draw=none] () {\smaller[2]{$\weightT(x,y)=x$}};
    \node[node distance=5mm,below of=out3,xshift=2mm,draw=none] () {\smaller[2]{$\weightT(x,y)=x$}};
  \end{scope}
\end{tikzpicture}
  \caption{The kernel $\Kernel$ (with input state $\loc_1$),
  and a semi-unfolding \tgame such that $\Val_\game(\loc_0,(0,0))=\Val_\tgame(\loc_0,(0,0))$.
  We denote $\loc_i$, $\loc_i'$ and $\loc_i''$ the locations in $\Kernel$, $\Kernel'$ and $\Kernel''$.}
  \label{fig:example-unfolding}
\end{figure}
Let us now compute an approximation of $\Val_\tgame$.
Let us first remove the states of value $+\infty$: $\loc_1'''$ and $\loc_2''$.
Then, we start at the bottom and compute an $(\varepsilon/3)$-approximation
of the value of $\loc_1''$ in the game defined by
$\Kernel_{\loc_1}''$ and its output transition to $\loc_t''$.
Following Section~\ref{sec:approx-kernels}, we should use $N\geq 3(4+1)/\varepsilon$
and compute values in the $1/N$-corners game ${\cal C}_{N}(\Kernel_{\loc_1}'')$ in order to obtain
an $(\varepsilon/3)$-approximation of the value function.
For $\varepsilon=15$ we will use $N=1$ (in this case the computation happens
to be exact and would also hold with a small $\varepsilon$)
We construct this corner game, and obtain the finite (untimed) weighted game in Figure~\ref{fig:example-corner-game}.

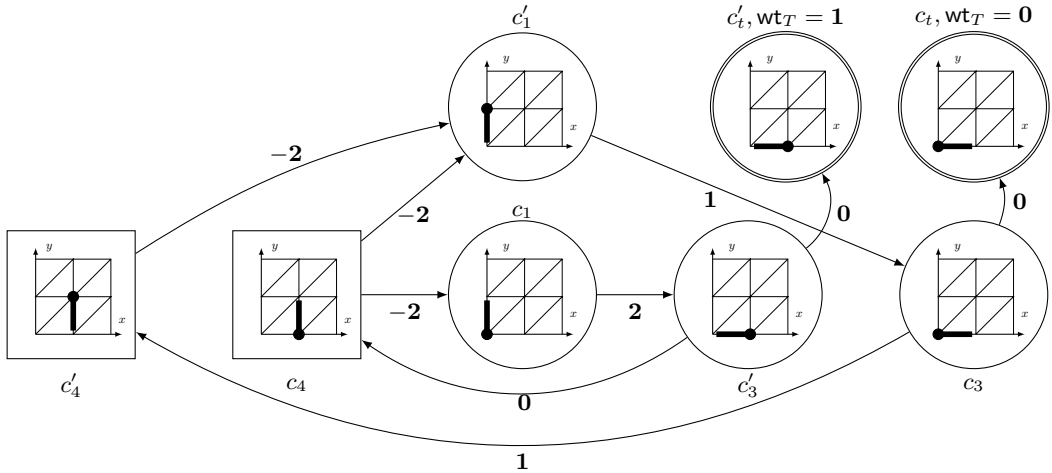
\begin{figure}[ht]
\centering
\begin{tikzpicture}[node distance=3cm,auto,->,>=latex]

  \node[player1,minimum size=10mm](1top){\makebox[0mm][c]{
  \scalebox{0.5}{\begin{tikzpicture}
    \path[draw,->,thick](0,0) -> (2.3,0) node[above] {$x$};
    \path[draw,->,thick](0,0) -> (0,2.3) node[right] {$y$};
    \path[draw,-] (0,0) -- (2,2)
    (1,0) -- (1,2) -- (0,1) -- (2,1) -- (1,0)
    (0,2) -- (2,2) -- (2,0);
    \node () at (1,-0.2) {};%
    \node () at (2,-0.2) {};%
    \node () at (-0.2,-0.2) {};%
    \node () at (-0.2,1) {};%
    \node () at (-0.2,2) {};%
    \path[draw, -, fill=black, line width=1.5mm] (0,0.1) -- (0,0.9);
    \node[draw, fill=black,circle,inner sep=1mm, minimum size=1mm] () at (0,1) {};
  \end{tikzpicture}}}};
  \node()[above of=1top,node distance=12mm]{$c_1'$};

  \node[player1,minimum size=10mm](1bot)[below of=1top,yshift=5mm]{\makebox[0mm][c]{
  \scalebox{0.5}{\begin{tikzpicture}
    \path[draw,->,thick](0,0) -> (2.3,0) node[above] {$x$};
    \path[draw,->,thick](0,0) -> (0,2.3) node[right] {$y$};
    \path[draw,-] (0,0) -- (2,2)
    (1,0) -- (1,2) -- (0,1) -- (2,1) -- (1,0)
    (0,2) -- (2,2) -- (2,0);
    \node () at (1,-0.2) {};%
    \node () at (2,-0.2) {};%
    \node () at (-0.2,-0.2) {};%
    \node () at (-0.2,1) {};%
    \node () at (-0.2,2) {};%
    \path[draw, -, fill=black, line width=1.5mm] (0,0.1) -- (0,0.9);
    \node[draw, fill=black,circle,inner sep=1mm, minimum size=1mm] () at (0,0) {};
  \end{tikzpicture}}}};
  \node()[above of=1bot,node distance=12mm,yshift=-1mm]{$c_1$};

  \node[player1,minimum size=10mm](3top)[right of=1bot]{\makebox[0mm][c]{
  \scalebox{0.5}{\begin{tikzpicture}
    \path[draw,->,thick](0,0) -> (2.3,0) node[above] {$x$};
    \path[draw,->,thick](0,0) -> (0,2.3) node[right] {$y$};
    \path[draw,-] (0,0) -- (2,2)
    (1,0) -- (1,2) -- (0,1) -- (2,1) -- (1,0)
    (0,2) -- (2,2) -- (2,0);
    \node () at (1,-0.2) {};%
    \node () at (2,-0.2) {};%
    \node () at (-0.2,-0.2) {};%
    \node () at (-0.2,1) {};%
    \node () at (-0.2,2) {};%
    \path[draw, -, fill=black, line width=1.5mm] (0.1,0) -- (0.9,0);
    \node[draw, fill=black,circle,inner sep=1mm, minimum size=1mm] () at (1,0) {};
  \end{tikzpicture}}}};
  \node()[below of=3top,node distance=12mm]{$c_3'$};

  \node[player1,minimum size=10mm](3bot)[right of=3top]{\makebox[0mm][c]{
  \scalebox{0.5}{\begin{tikzpicture}
    \path[draw,->,thick](0,0) -> (2.3,0) node[above] {$x$};
    \path[draw,->,thick](0,0) -> (0,2.3) node[right] {$y$};
    \path[draw,-] (0,0) -- (2,2)
    (1,0) -- (1,2) -- (0,1) -- (2,1) -- (1,0)
    (0,2) -- (2,2) -- (2,0);
    \node () at (1,-0.2) {};%
    \node () at (2,-0.2) {};%
    \node () at (-0.2,-0.2) {};%
    \node () at (-0.2,1) {};%
    \node () at (-0.2,2) {};%
    \path[draw, -, fill=black, line width=1.5mm] (0.1,0) -- (0.9,0);
    \node[draw, fill=black,circle,inner sep=1mm, minimum size=1mm] () at (0,0) {};
  \end{tikzpicture}}}};
  \node()[below of=3bot,node distance=12mm]{$c_3$};

  \node[draw,regular polygon,regular polygon sides=4,inner sep=-2mm,minimum size=10mm](4bot)[left of=1bot]{\makebox[0mm][c]{
  \scalebox{0.5}{\begin{tikzpicture}
    \path[draw,->,thick](0,0) -> (2.3,0) node[above] {$x$};
    \path[draw,->,thick](0,0) -> (0,2.3) node[right] {$y$};
    \path[draw,-] (0,0) -- (2,2)
    (1,0) -- (1,2) -- (0,1) -- (2,1) -- (1,0)
    (0,2) -- (2,2) -- (2,0);
    \node () at (1,-0.2) {};%
    \node () at (2,-0.2) {};%
    \node () at (-0.2,-0.2) {};%
    \node () at (-0.2,1) {};%
    \node () at (-0.2,2) {};%
    \path[draw, -, fill=black, line width=1.5mm] (1,0.1) -- (1,0.9);
    \node[draw, fill=black,circle,inner sep=1mm, minimum size=1mm] () at (1,0) {};
  \end{tikzpicture}}}};
  \node()[below of=4bot,node distance=12mm]{$c_4$};

  \node[draw,regular polygon,regular polygon sides=4,inner sep=-2mm,minimum size=10mm](4top)[left of=4bot]{\makebox[0mm][c]{
  \scalebox{0.5}{\begin{tikzpicture}
    \path[draw,->,thick](0,0) -> (2.3,0) node[above] {$x$};
    \path[draw,->,thick](0,0) -> (0,2.3) node[right] {$y$};
    \path[draw,-] (0,0) -- (2,2)
    (1,0) -- (1,2) -- (0,1) -- (2,1) -- (1,0)
    (0,2) -- (2,2) -- (2,0);
    \node () at (1,-0.2) {};%
    \node () at (2,-0.2) {};%
    \node () at (-0.2,-0.2) {};%
    \node () at (-0.2,1) {};%
    \node () at (-0.2,2) {};%
    \path[draw, -, fill=black, line width=1.5mm] (1,0.1) -- (1,0.9);
    \node[draw, fill=black,circle,inner sep=1mm, minimum size=1mm] () at (1,1) {};
  \end{tikzpicture}}}};
  \node()[below of=4top,node distance=12mm]{$c_4'$};

  \node[player1,minimum size=10mm](5top)[right of=1top,xshift=5mm,accepting]{\makebox[0mm][c]{
  \scalebox{0.5}{\begin{tikzpicture}
    \path[draw,->,thick](0,0) -> (2.3,0) node[above] {$x$};
    \path[draw,->,thick](0,0) -> (0,2.3) node[right] {$y$};
    \path[draw,-] (0,0) -- (2,2)
    (1,0) -- (1,2) -- (0,1) -- (2,1) -- (1,0)
    (0,2) -- (2,2) -- (2,0);
    \node () at (1,-0.2) {};%
    \node () at (2,-0.2) {};%
    \node () at (-0.2,-0.2) {};%
    \node () at (-0.2,1) {};%
    \node () at (-0.2,2) {};%
    \path[draw, -, fill=black, line width=1.5mm] (0.1,0) -- (0.9,0);
    \node[draw, fill=black,circle,inner sep=1mm, minimum size=1mm] () at (1,0) {};
  \end{tikzpicture}}}};
  \node()[above of=5top,node distance=12mm]{$c_t', \weightT=\mathbf{1}$};

  \node[player1,minimum size=10mm](5bot)[right of=5top,xshift=-5mm,accepting]{\makebox[0mm][c]{
  \scalebox{0.5}{\begin{tikzpicture}
    \path[draw,->,thick](0,0) -> (2.3,0) node[above] {$x$};
    \path[draw,->,thick](0,0) -> (0,2.3) node[right] {$y$};
    \path[draw,-] (0,0) -- (2,2)
    (1,0) -- (1,2) -- (0,1) -- (2,1) -- (1,0)
    (0,2) -- (2,2) -- (2,0);
    \node () at (1,-0.2) {};%
    \node () at (2,-0.2) {};%
    \node () at (-0.2,-0.2) {};%
    \node () at (-0.2,1) {};%
    \node () at (-0.2,2) {};%
    \path[draw, -, fill=black, line width=1.5mm] (0.1,0) -- (0.9,0);
    \node[draw, fill=black,circle,inner sep=1mm, minimum size=1mm] () at (0,0) {};
  \end{tikzpicture}}}};
  \node()[above of=5bot,node distance=12mm]{$c_t, \weightT=\mathbf{0}$};

  \path
  (1bot) edge node[below]{$\mathbf{2}$} (3top)
  (3top) edge[bend left=35] node[below,yshift=1mm]{$\mathbf{0}$} (4bot)
  (4bot) edge node[below]{$\mathbf{-2}$} (1bot)
  (1top) edge node[above,xshift=-5mm,yshift=-2mm]{$\mathbf{1}$} (3bot)
  (3bot) edge[bend left=30] node[below]{$\mathbf{1}$} (4top)
  (4top) edge[bend left=10] node[above]{$\mathbf{-2}$} (1top)
  (4bot) edge node[below]{$\mathbf{-2}$} (1top)
  (3top) edge[bend right=40] node[right]{$\mathbf{0}$} (5top)
  (3bot) edge[bend right=20] node[right]{$\mathbf{0}$} (5bot)
  ;

\end{tikzpicture}
  \caption{
  The finite weighted game obtained from ${\cal C}_{1}(\Kernel_{\loc_1}'')$,
  where $c_i$ and $c_i'$ are the corners of $\loc_i''$ in \tgame.}
  \label{fig:example-corner-game}
\end{figure}

We can compute the values in this game to obtain $\Val(c_1')=1$ and $\Val(c_1)=3$.
We then define a value for every configuration in state $\loc_1''$ by linear interpolation,
obtaining $(x,y)\to 3-2y$ (which happens to be exactly $(x,y)\to\Val_{\tgame}(\loc_1'',(x,y))$
in this case, but would only be an $\varepsilon/3$-approximation of it in general).
Now, we can compute an $\varepsilon/3$-approximation of $\Val_{\tgame}(\loc_2')$
with one step of value iteration, obtaining $(x,y)\to\inf_{0<d<2-x} (-1)*d+1+3-2(0+d)=3x-2$.

The next step is computing an $\varepsilon/3$-approximation of the value of $\loc_1'$ in the game defined by
$\Kernel_{\loc_1}'$ and its output transitions to $\loc_t'$ and $\loc_2'$, of
respective output weight functions $(x,y)\to x$ and $(x,y)\to 3x-2$.
This will give us an $2\varepsilon/3$-approximation of $\Val_{\tgame}(\loc_1')$.

Following Section~\ref{sec:approx-kernels} once again, we should use $N\geq 3(5+3)/\varepsilon$
and compute values in the $1/N$-corners game ${\cal C}_{N}(\Kernel_{\loc_1}')$.
For $\varepsilon=15$ this gives $N=2$ (which will once again keep the computation exact).
We can construct a finite (untimed) weighted game as in Figure~\ref{fig:example-corner-game},
and obtain a value for each $1/2$-corner of state $\loc_1'$:
On the $1/2$-region $(0<y<1/2,x=0)$, corner $(0,0)$ has value $2$ and corner $(0,1/2)$
has value $2$.
On the $1/2$-region $(y=1/2,x=0)$, corner $(0,1/2)$ has value $2$.
On the $1/2$-region $(1/2<y<1,x=0)$, corner $(0,1/2)$ has value $2$ and corner $(0,1)$
has value $1$.
From these results, we define a piecewise-linear function by interpolating
the values of corners on each $1/2$-region, and obtain
$(x,y)\to\begin{cases} 2 & \text{if } y\leq 1/2 \\ 3-2y & \text{otherwise}\end{cases}$,
as depicted in Figure~\ref{fig:example-value}.

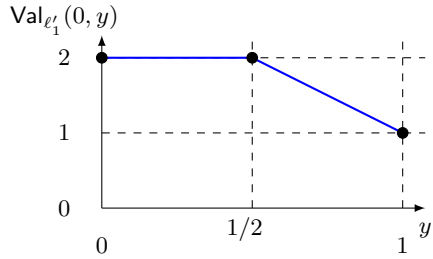
\begin{figure}[ht]
\centering
\begin{tikzpicture}[node distance=3cm,auto,>=latex]

 \draw[->] (0,0) -- (0,2.3);
 \draw[->] (0,0) -- (4.3,0);

 \draw[dashed] (4,0) -- (4,2.3);

 \draw[dashed] (0,1) -- (4.3,1);
 \draw[dashed] (0,2) -- (4.3,2);

 \draw[dashed] (2,0) node[below,xshift=-1mm]{$1/2$}-- (2,2.3);

 \node at (4.3,-.3) {$y$};
 \node at (-.5,2.5) {$\Val_{\loc_1'}(0,y)$};

 \node at (0,-.5) {$0$};
 \node at (4,-.5) {$1$};

 \node at (-.5,0) {$0$};
 \node at (-.5,1) {$1$};
 \node at (-.5,2) {$2$};

 \draw[blue,thick] (0,2) -- (2,2) -- (4,1);
 \node () at (0,2) {};
 \node[draw, fill=black,circle,inner sep=0.5mm] () at (0,2) {};
 \node[draw, fill=black,circle,inner sep=0.5mm] () at (2,2) {};
 \node[draw, fill=black,circle,inner sep=0.5mm] () at (4,1) {};

\end{tikzpicture}
  \caption{
  The value function $(x,y)\to\Val_{\tgame}(\loc_1',(x,y))$, projected on $x=0$.
  Black dots represent the values obtained for $1/2$-corners using the corner-points
  abstraction.}
  \label{fig:example-value}
\end{figure}

This gives us an $2\varepsilon/3$-approximation of $(x,y)\to\Val_{\tgame}(\loc_1',(x,y))$
(in fact exactly $\Val_{\tgame}(\loc_1')$).
Now, we can compute an $2\varepsilon/3$-approximation of $\Val_{\tgame}(\loc_2)$
on region $(1<x<2, y=0)$ with one step of value iteration,
obtaining :
\[(x,y)\to  \inf_{0<d<2-x}
\begin{cases} 3-d & \text{if } d\leq 1/2 \\
   4-3d & \text{otherwise}
\end{cases} \allowbreak=
\begin{cases} 3x-2 & \text{if } x\leq 3/2 \\ x+1 & \text{otherwise}
\end{cases}\]

Then, we need to compute an $\varepsilon/3$-approximation of the value of $\loc_1$ in the game defined by
$\Kernel_{\loc_1}$ and its output transitions to $\loc_t$ and $\loc_2$, of
respective output weight functions $(x,y)\to x$ and
$(x,y)\to 3x-2 \text{ if } x\leq 3/2, x+1 \text{ otherwise}$.
This will give us an $\varepsilon$-approximation of $\Val_{\tgame}(\loc_1)$.

Following Section~\ref{sec:approx-kernels} one last time, we should use $N\geq 3(5+3)/\varepsilon$
and compute values in the $1/N$-corners game ${\cal C}_{N}(\Kernel_{\loc_1})$.
This time, let us use $N=3$ to showcase an example where the computed value is not exact.
We can construct a finite (untimed) weighted game as in Figure~\ref{fig:example-corner-game},
and obtain a value for each $1/3$-corner of state $\loc_1'$.
From these results, we define a piecewise-linear function by interpolation,
as depicted in Figure~\ref{fig:example-value-approx}.

\begin{figure}[ht]
\centering
\begin{tikzpicture}[node distance=3cm,auto,>=latex]

 \draw[->] (0,0) -- (0,2.3);
 \draw[->] (0,0) -- (4.3,0);

 \draw[dashed] (4,0) -- (4,2.3);

 \draw[dashed] (0,1) -- (4.3,1);
 \draw[dashed] (0,2) -- (4.3,2);

 \draw[dashed] (1.333,0) node[below,xshift=-1mm]{$1/3$}-- (1.333,2.3);
  \draw[dashed] (2.666,0) node[below,xshift=-1mm]{$2/3$}-- (2.666,2.3);

 \node at (4.3,-.3) {$y$};
 \node at (-.5,2.5) {$\Val_{\loc_1}(0,y)$};

 \node at (0,-.5) {$0$};
 \node at (4,-.5) {$1$};

 \node at (-.5,0) {$0$};
 \node at (-.5,1) {$1$};
 \node at (-.5,2) {$2$};

 \draw[red,thick] (0,2) -- (2,2) -- (4,1);
 \draw[blue,thick] (0,2) -- (1.333,2) -- (2.666,1.666) -- (4,1);
 \node () at (0,2) {};
 \node[draw, fill=black,circle,inner sep=0.5mm] () at (0,2) {};
 \node[draw, fill=black,circle,inner sep=0.5mm] () at (1.333,2) {};
 \node[draw, fill=black,circle,inner sep=0.5mm] () at (2.666,1.666) {};
 \node[draw, fill=black,circle,inner sep=0.5mm] () at (4,1) {};

\end{tikzpicture}
  \caption{
  The value function $(x,y)\to\Val_{\tgame}(\loc_1,(x,y))$, projected on $x=0$,
  is depicted in red.
  Black dots represent the values obtained for $1/3$-corners using the corner-points
  abstraction, and the derived approximation of the value function is depicted in blue}
  \label{fig:example-value-approx}
\end{figure}
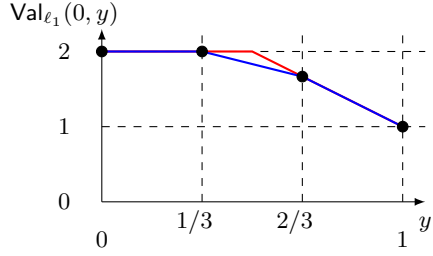

Finally, from this $\varepsilon$-approximation of $\Val_{\tgame}(\loc_1)$,
we can compute an $\varepsilon$-approximation of $\Val_{\tgame}(\loc_0)$
using one step of value iteration, and conclude. On our example this
ensures $$\Val_{\tgame}(\loc_0,(0,0))=\sup_{0<d<1} \Val_{\tgame}(\loc_1,(0,d))
\in[2-\varepsilon,2+\varepsilon]$$.

\subsection{Complexity analysis}\label{app:complexity}

We will express complexities according to several parameters: $|\Locs|$, $|\Clocks|$,
greatest guard constant $\clockbound$, greatest location and transition weight
constants $\wmax^\Locs$ and $\wmax^\Trans$.
We also need to keep track of the output weight functions' characteristics.
Recall that the output weight functions must be piecewise linear with finitely many pieces
and Lipschitz-continuous over regions.
We define three parameters, its Lipschitz constant $\lipconst$, its number of linear pieces $\numpieces$
and a bound $\addconstbound$ (that we call additive bound) on its additive constant, such that
if $(x_1,\dots,x_{|\Clocks|})\to \sum_{i=1}^{|\Clocks|} a_i x_i + b$
defines one of those linear pieces, then $|b|\leq \addconstbound$ and $\forall 1\leq i\leq |\Clocks|, |a_i|\leq \lipconst$.

Note that $|\Locs|$, $|\Clocks|$ and $\numpieces$ are all polynomial in the size of the input,
but \clockbound, $\wmax^\Locs$, $\wmax^\Trans$, $\lipconst$ and $\addconstbound$ are exponential in the size of the input if
constants are encoded in binary.

We start with simple estimates:
\begin{itemize}
  \item Number of regions $|\regions\Clocks\clockbound|$:
        Polynomial in \clockbound, exponential in $|\Clocks|$.
  \item Number of $1/N$-regions $|\Nregions N\Clocks\clockbound|$:
        Polynomial in \clockbound and $N$, exponential in $|\Clocks|$.
  \item Number of $1/N$-corners:
        Polynomial in \clockbound and $N$, exponential in $|\Clocks|$.
  \item Maximum weight of a timed transition $\wmax^e$:
        Polynomial in $\clockbound$, $\wmax^\Locs$ and $\wmax^\Trans$.
  \item Maximum output weight $\sup|\weightT|$:
        Polynomial in \clockbound, \addconstbound, $|\Clocks|$ and \lipconst.
\end{itemize}

\subsubsection{Tree}

Let us recall the complexity of the value iteration algorithm, used to compute the exact value of an acyclic WTG:
\\
Input: An acyclic game of depth $i$.
\\
Algorithm schema: Computes $\mathcal{F}^i(V^0)=\Val^i=\Val$.
\\
Output: A $\lipconst'$-Lipschitz-continuous function with $\numpieces'$ pieces and additive bound $\addconstbound'$ that is the game's value.
  \begin{itemize}
  \item $\lipconst'$ is of the form $\lipconst\,\lipconst''$ with $\lipconst''$
  polynomial in $\wmax^\Locs$ and $|\Clocks|$
  and exponential in $i$.
  \item $\numpieces'$ is of the form $\numpieces^{|\Clocks|}\,\numpieces''$ with $\numpieces''$
  polynomial in \clockbound and $|\Locs|$
  and exponential in $|\Clocks|$ and $i$.
  \item $\addconstbound'$ is of the form $\addconstbound+\addconstbound''$ with $\addconstbound''$
  polynomial in \clockbound, $\wmax^\Locs$, $\wmax^\Trans$ and $i$.
  \end{itemize}

Complexity: exponential in $i$ and the size of the input.

\subsubsection{Kernel}

Input: A kernel WTG, a precision $\varepsilon>0$.
\\
Algorithm schema: Solves optimal reachability on the finite $1/N$-corner game with $N$ polynomial in $1/\varepsilon$,
  $\wmax^\Locs$, $|\Locs|$, \clockbound and \lipconst
  and exponential in $|\Clocks|$.
\\
Output: A $\lipconst'$-Lipschitz-continuous value function with $\numpieces'$ pieces and additive bound $\addconstbound'$ that is an $\varepsilon$-approximation of the game's value.
  \begin{itemize}
  \item $\lipconst'$ is of the form $\lipconst\,\lipconst''$ with $\lipconst''$
  polynomial in $|\Locs|$, $\wmax^\Locs$ and \clockbound
  and exponential in $|\Clocks|$.
  \item $\numpieces'$ is polynomial in $1/\varepsilon$,
  $\wmax^\Locs$, $|\Locs|$, \clockbound and \lipconst
  and exponential in $|\Clocks|$ (in particular, it is independent in \numpieces).
  \item $\addconstbound'$ of the form $\addconstbound+\addconstbound''$ with $\addconstbound''$
  polynomial in $1/\varepsilon$,
  $\wmax^\Locs$, $\wmax^\Trans$, $|\Locs|$, \clockbound and \lipconst
  and exponential in $|\Clocks|$.
  \end{itemize}

Complexity: polynomial in $1/\varepsilon$,
  $\wmax^\Locs$, $\wmax^\Trans$, $|\Locs|$, \clockbound and \lipconst
  and exponential in $|\Clocks|$.

\subsubsection{Semi-unfolding}

We now stack several kernel and tree parts to form a semi-unfolding of a region game.
\\
Input: A semi-unfolding of branch depth $D$, a precision $\varepsilon>0$.
\\
Algorithm schema: value iteration for the trees and region-based for the kernels
  (on $1/N$ corners), with precision $\varepsilon/D$.
  In order to bound $N$, we need to bound the Lipschitz constants along the whole computation.
  We can recursively show that along this computation the Lipschitz constants, additive constants and number of pieces
  do not grow too much, and obtain global bounds:
  \begin{itemize}
  \item we can bound all Lipschitz constants by
  $\lipconst\,\lipconst''$ with $\lipconst''$ polynomial in $|\Locs|$, $\wmax^\Locs$, \clockbound
  and exponential in $|\Clocks|$ and $D$.
  \item we can bound all number of pieces by $\numpieces^{|\Clocks|}\,\numpieces''$ with $\numpieces''$
  polynomial in $1/\varepsilon$, \clockbound, $|\Locs|$, $\wmax^\Locs$, and \lipconst
  and exponential in $|\Clocks|$ and $D$.
  \item we can bound all additive constants by $\addconstbound+\addconstbound''$ with $\addconstbound''$
  polynomial in $1/\varepsilon$,
  $\wmax^\Locs$, $\wmax^\Trans$, $|\Locs|$, \clockbound and \lipconst
  and exponential in $|\Clocks|$ and $D$.
  \end{itemize}
  Therefore, $N$ can be chosen polynomial in $1/\varepsilon$,
  $\wmax^\Locs$, $|\Locs|$, \clockbound and $\lipconst$
  and exponential in $|\Clocks|$ and $D$.
\\
Output: A $\lipconst'$-Lipschitz-continuous value function with $\numpieces'$ pieces
and additive bound $\addconstbound'$ that is an $\varepsilon$-approximation of the game's value.
$\lipconst'$, $\numpieces'$, $\addconstbound'$ are bounded by their respective global bound.
\\
Complexity: polynomial in $1/\varepsilon$ and exponential in the size of the input and $D$.

\subsubsection{Almost divergent game}

Input: An almost divergent game, a precision $\varepsilon>0$.
\\
Algorithm schema: First, compute the region game's SCCs, and remove $+-\infty$ locations.
  Then, perform the semi-unfolding of the game, of depth $D$
  whose value is equivalent to that of the original game, with $D$ polynomial in
  \clockbound, $|\Locs|$, $\wmax^\Locs$, $\wmax^\Trans$, \lipconst, \addconstbound
  and exponential in $|\Clocks|$.
\\
Output: A $\lipconst'$-Lipschitz-continuous value function with $\numpieces'$ pieces and additive bound $\addconstbound'$ that is an $\varepsilon$-approximation of the game's value.
  \begin{itemize}
  \item $\lipconst'$ is
  exponential in \clockbound, $|\Locs|$, $\wmax^\Locs$, $\wmax^\Trans$, \lipconst, \addconstbound
  and doubly-exponential in $|\Clocks|$.
  \item $\numpieces'$ is polynomial in \numpieces, $1/\varepsilon$,
  exponential in \clockbound, $|\Locs|$, $\wmax^\Locs$, $\wmax^\Trans$, \lipconst, \addconstbound
  and doubly-exponential in $|\Clocks|$.
  \item $\addconstbound'$ is polynomial in $1/\varepsilon$
  and exponential in \clockbound, $|\Locs|$, $\wmax^\Locs$, $\wmax^\Trans$, \lipconst, \addconstbound
  and doubly-exponential in $|\Clocks|$
  \end{itemize}

Complexity: polynomial in $1/\varepsilon$,
  exponential in the size of the input and \clockbound, \lipconst, \addconstbound, $\wmax^\Locs$ and $\wmax^\Trans$
  and doubly-exponential in $|\Clocks|$.

\section{Proofs of the symbolic approximation schema  (Section~\ref{sec:symbolic})}\label{app:symbolic}

\begin{proof}[Proof of Lemma~\ref{lm:symbolic-scc-plus}]
Consider a non-negative SCC's \game, a precision $\varepsilon$, and an initial configuration $(\loc_0,\val_0)$.
Let \tgame be its finite semi-unfolding (obtained from the labelled tree $T$,
as in Appendix~\ref{app:unfolding}), such that
$\Val_\game(\loc_0,\val_0)=\Val_\tgame((\tilde{\loc}_0,r_0),\val_0)$.
Let $\alpha$ be the maximum number of kernels along a branch of $T$.
Let $P'$ be an integer such that for all kernels \Kernel in \tgame,
$|\Val_\Kernel(\loc,\val)-\Val^{P'}_\Kernel(\loc,\val)|\leq \varepsilon/\alpha$
for all configurations $(\loc,\val)$ of $\game$.
We can find such a $P'$ by using Lemma~\ref{lm:symbolic-kernel}.

Create $\mathcal T'(\game)$ from $T$ by applying the method used to create \tgame
but replace every kernel by its complete $P'$-unfolding instead.
This implies that $\mathcal T'(\game)$ is a tree, of bounded depth $P$ (at most the depth of $T$ times $P'$).
Then $|\Val_\tgame((\tilde{\loc}_0,r_0),\val_0)-
\Val_{\mathcal T'(\game)}((\tilde{\loc}_0,r_0),\val_0)|\leq\varepsilon$.
This holds because the value function is $1$-Lipschitz-continuous with regards to the
output weight function, so imprecision builds up additively.

Consider now $\mathcal T''(\game)$ the (complete) unfolding of \rgame with
unfolding depth $P$, where kernels are also unfolded.
By construction, $\Val_{\mathcal T''(\game)}((\tilde{\loc}_0,r_0),\val_0)
=\Val_{\mathcal T''(\game)}^P((\tilde{\loc}_0,r_0),\val_0)$.
Then, we can prove that $\Val_{\mathcal T''(\game)}^P((\tilde{\loc}_0,r_0),\val_0)
=\Val_\game^P(\loc_0,\val_0)$ (same strategies at bounded horizon $P$),
which implies $\Val_{\rgame)}((\loc_0,r_0),\val_0)\leq
\Val_{\mathcal T''(\game)}((\tilde{\loc}_0,r_0),\val_0)$ (monotonicity of $\Val^k$).
By another monotonicity argument (because $\mathcal T''$ contains $\mathcal T'$ as a prefix),
we can also prove $\Val_{\mathcal T''(\game)}((\tilde{\loc}_0,r_0),\val_0)\leq
\Val_{\mathcal T'(\game)}((\tilde{\loc}_0,r_0),\val_0)$.

Bringing everything together we obtain
$|\Val_\game^P(\loc_0,\val_0)-\Val_\game(\loc_0,\val_0)|\leq\varepsilon$.
\end{proof}

\begin{proof}[Proof of Lemma~\ref{lm:symbolic-scc-minus}]
Consider a non-positive SCC \game, a precision $\varepsilon$, and an initial configuration $(\loc_0,\val_0)$.
Let \tgame be its finite semi-unfolding (obtained from the labelled tree $T$,
as in Appendix~\ref{app:unfolding}), such that $\Val_\game(\loc_0,\val_0)=\Val_\tgame((\tilde{\loc}_0,r_0),\val_0)$.

We now change $T$, by adding a subtree under each stopped leaf: the complete unfolding of \rgame,
starting from the stopped leaf, of depth $|\rgame|$. Let us name $T^+$ this unfolding tree.
We then construct $\mathcal T^+(\game)$ as before, based on $T^+$.
Since we are in a non-positive SCC, $\mathcal T^+(\game)$ must have
output weight $-\infty$ on its stopped leaves.
It is easy to see that $\Val_\game(\loc_0,\val_0)=\Val_{\mathcal T^+(\game)}((\tilde{\loc}_0,r_0),\val_0)$
still holds (the proof was based on branches being long enough, and we
increased the lengths).
We now perform a small but crucial change: the output weight of stopped leaves
in $\mathcal T^+(\game)$ is set to $+\infty$ instead of $-\infty$.
Trivially $\Val_\tgame((\tilde{\loc}_0,r_0),\val_0)\leq
\Val_{\mathcal T^+(\game)}((\tilde{\loc}_0,r_0),\val_0)$ (we increased the output weight function).
Let us prove that $\Val_{\mathcal T^+(\game)}((\tilde{\loc}_0,r_0),\val_0)\leq
\Val_\tgame((\tilde{\loc}_0,r_0),\val_0)$.

For a fixed $\eta>0$, consider $\stratmin$ a $\eta$-optimal strategy for player \MinPl
in \tgame.
Let us define $\stratmin^+$, a strategy for \MinPl in $\mathcal T^+(\game)$, by
making the same choice as \stratmin on the common prefix tree, and once a node
that is a stopped leaf in \tgame is reached, we switch to a memoryless attractor
strategy of \MinPl towards target states.
Consider any strategy $\stratmax^+$ of \MaxPl in $\mathcal T^+(\game)$,
and let \stratmax be its projection in \tgame.
Let $\play^+$ denote the (maximal) play $\outcomes_{\mathcal T^+(\game)}(((\loc_0,r_0),\val_0),\stratmin^+,\stratmax^+))$,
and $\play$ be $\outcomes_{\tgame}(((\loc_0,r_0),\val_0),\stratmin,\stratmax))$.
By construction, $\play^+$ does not reach a stopped leaf in $\mathcal T^+(\game)$.
If the play $\play^+$
stays in the common prefix tree of $T$ and $T^+$, then $\play=\play^+$,
and $\weight_{\mathcal T^+(\game)}(\play^+)\leq \Val_\tgame((\tilde{\loc}_0,r_0),\val_0)+\eta$.
If it doesn't, then $\play^+$ has a prefix that reaches a stopped leaf
in \tgame: this must be $\rho$.
This implies that $\weight_{\mathcal T^+(\game)}(\play^+)<-|\rgame|\wmax^e-\sup|\weightT|\leq\Val_\tgame((\tilde{\loc}_0,r_0),\val_0)$
 (see Lemma~\ref{lm:mimickedplays}).
Since this holds for all $\eta>0$, we proved $\Val_{\mathcal T^+(\game)}((\tilde{\loc}_0,r_0),\val_0)\leq
\Val_\tgame((\tilde{\loc}_0,r_0),\val_0)$, which finally implies that
the two values are equal.

Then, we can follow the proof of Lemma~\ref{lm:symbolic-scc-plus}
(with $T^+$ and $\mathcal T^+(\game)$) in order to conclude.
\end{proof}

\end{document}